\documentclass{vldb}
\usepackage{times}
\pdfoutput=1
\usepackage[show]{chato-notes}
\usepackage{graphicx}
\usepackage{color}
\usepackage{cite}
\usepackage{url}
\usepackage{amsmath,amssymb}
\usepackage{dsfont}
\usepackage{epsfig}
\usepackage{framed}
\usepackage[ruled,lined,linesnumbered,noend]{algorithm2e}
\usepackage{multirow}
\usepackage{capt-of}
\usepackage{textcomp}
\usepackage{tabularx}
\usepackage{subfigure}
\usepackage[normalem]{ulem}

\usepackage{hyperref}

\newcommand{\blue}[1]{\textcolor{black}{#1}}

\newcommand{\spara}[1]{\smallskip\noindent{\bf #1}}

\newtheorem{proposition}{Proposition}

\newtheorem{theorem}{Theorem}

\newtheorem{claim}{Claim}
\newtheorem{example}{Example}
\newtheorem{lemma}{Lemma}
\newtheorem{problem}{Problem}

\newcommand{\LL}[1]{\textcolor{black}{#1}} 
\newcommand{\CA}[1]{\textcolor{black}{#1}}
\newcommand{\WL}[1]{\textcolor{black}{#1}}

\newcommand{\I}{\ensuremath{\mathcal{I}}}

\DeclareMathOperator*{\argmin}{arg\,min}
\DeclareMathOperator*{\argmax}{arg\,max}

\newcommand{\NPhard}{NP-hard\xspace}

\newcommand{\SPhard}{\#P-hard\xspace}

\newcommand{\MRA}{\textsc{Regret-Minimization}\xspace}

\newcommand{\defaultprob}{\delta}

\newcommand{\cals}{\mbox {${\cal S}$}}
\newcommand{\cala}{\mbox {${\cal A}$}}
\newcommand{\calr}{\mbox {${\cal R}$}}
\newcommand{\cali}{${\cal I}$\xspace}
\newcommand{\calj}{${\cal J}$\xspace}
\newcommand{\calb}{\mbox {${\cal B}$}}

\newcommand{\sigic}{\mbox {$\sigma^{ic}$}} 
\newcommand{\sigicctp}{\mbox {$\sigma^{icctp}$}}

\newcommand{\regret}{\mathcal{R}}

\newcommand{\E}{\mathrm{E}}
\newcommand{\RR}{\mathbf{R}}
\newcommand{\RQ}{\mathbf{Q}}

\newcommand{\fastAlgorithm}{Two-phase Iterative Regret Minimization\xspace}
\newcommand{\fastAlgo}{\textsc{Tirm}\xspace}

\newcommand{\eat}[1]{}

\newcommand{\squishlist}{
 \begin{list}{$\bullet$}
  {  \setlength{\itemsep}{0pt}
     \setlength{\parsep}{3pt}
     \setlength{\topsep}{3pt}
     \setlength{\partopsep}{0pt}
     \setlength{\leftmargin}{2em}
     \setlength{\labelwidth}{1.5em}
     \setlength{\labelsep}{0.5em}
} }
\newcommand{\squishlisttight}{
 \begin{list}{$\bullet$}
  { \setlength{\itemsep}{0pt}
    \setlength{\parsep}{0pt}
    \setlength{\topsep}{0pt}
    \setlength{\partopsep}{0pt}
    \setlength{\leftmargin}{2em}
    \setlength{\labelwidth}{1.5em}
    \setlength{\labelsep}{0.5em}
} }

\newcommand{\squishdesc}{
 \begin{list}{}
  {  \setlength{\itemsep}{0pt}
     \setlength{\parsep}{3pt}
     \setlength{\topsep}{3pt}
     \setlength{\partopsep}{0pt}
     \setlength{\leftmargin}{1em}
     \setlength{\labelwidth}{1.5em}
     \setlength{\labelsep}{0.5em}
} }

\newcommand{\squishend}{
  \end{list}
}

\newcommand{\flix}{\textsc{Flixster}\xspace}
\newcommand{\livej}{\textsc{LiveJournal}\xspace}

\newcommand{\irie}{\textsc{Greedy-Irie}\xspace}
\newcommand{\myopicOne}{\textsc{Myopic}\xspace}
\newcommand{\myopicTwo}{\textsc{Myopic+}\xspace}

\newcommand{\epi}{\textsc{Epinions}\xspace}
\newcommand{\dblp}{\textsc{DBLP}\xspace}

\begin{document}
\title{Viral Marketing Meets Social Advertising:\\ Ad Allocation with Minimum Regret}

\author{\begin{tabular}{ccccc}
Cigdem Aslay$^{1,2}$ & Wei Lu$^3$ & Francesco Bonchi$^2$ & Amit Goyal$^4$ & Laks V.S. Lakshmanan$^3$
\end{tabular}
\\
\\
\begin{tabular}{cccc}
  $^1$Univ. Pompeu Fabra &  $^2$Yahoo Labs &  $^3$Univ. of British Columbia &  $^4$Twitter \\
 Barcelona, Spain & Barcelona, Spain & Vancouver, Canada & San Francisco, CA \\
 {\sf aslayci@acm.org} & {\sf bonchi@yahoo-inc.com} & {\sf \{welu,laks\}@cs.ubc.ca} & {\sf agoyal@twitter.com}
\end{tabular}
}

\maketitle

\enlargethispage*{\baselineskip}
\begin{abstract}
Social advertisement is one of the fastest growing sectors in the digital advertisement landscape: ads in the form of promoted posts are shown in the feed of users of a social networking platform, along with normal social posts; if a user clicks on a promoted post, the host (social network owner) is paid a fixed amount from the advertiser. In this context, allocating ads to users is typically performed by maximizing click-through-rate, i.e., the likelihood that the user will click on the ad. However, this simple strategy fails to leverage the fact the ads can propagate virally through the network, from endorsing users to their followers.

In this paper, we study the problem of allocating ads to users through the viral-marketing lens. Advertisers approach the host with a budget in return for the marketing campaign service provided by the host. We show that allocation that takes into account the propensity of ads for viral propagation can achieve significantly better performance. However, uncontrolled virality could be undesirable for the host as it creates room for exploitation by the advertisers: hoping to tap uncontrolled virality, an advertiser might declare a lower budget for its marketing campaign, aiming at the same large outcome with a smaller cost.

This creates a challenging trade-off: on the one hand, the host aims at leveraging virality and the network effect to improve advertising efficacy, while on the other hand the host wants to avoid giving away free service due to uncontrolled virality.
We formalize this as the problem of ad allocation with minimum regret, which we show is \NPhard and inapproximable w.r.t. any factor. However, we devise an algorithm that provides approximation guarantees w.r.t. the total budget of all advertisers. We develop a scalable version of our approximation algorithm, which we extensively test on four real-world data sets, confirming that our algorithm delivers high quality solutions, is scalable, and significantly outperforms several natural baselines.
\end{abstract}

\section{Introduction}
\label{sec:intro}
\enlargethispage*{\baselineskip}

Advertising on social networking and microblogging platforms is one of
the fastest growing sectors in digital advertising, further fueled by the
explosion of investments in mobile ads.  Social ads are
typically implemented by platforms such as \texttt{Twitter},
\texttt{Tumblr}, and \texttt{Facebook} through the mechanism of
\emph{promoted posts} shown in the ``timeline'' (or feed) of their
users. A promoted post can be a video, an image, or simply a textual
post containing an advertising message.  Similar to organic (non-promoted) posts,
promoted posts can propagate from user to user in the network by means
of social actions such as \emph{``likes''}, \emph{``shares''}, or
\emph{``reposts''}.\footnote{\small \texttt{Tumblr}'s CEO David Karp reported (CES 2014) that a normal post is reposted
on average 14 times, while promoted posts are on average reposted more
than 10\,000 times:
\url{http://yhoo.it/1vFfIAc}.}
Below, we blur the distinction between these different types of action,
and generically refer to them all as \emph{clicks}. These actions have
two important aspects in common: (1) they can be seen as an
explicit form of acceptance or endorsement of the advertising message;
(2) they allow the promoted posts to propagate, so that they might be visible to
the \emph{``friends''} or \emph{``followers''} of the endorsing (i.e.,
clicking) users. In particular, the \LL{system} may supplement the ads with
\emph{social proofs} such as \emph{``X, Y, and 3 other friends clicked
on it''}, which may further increase the chance that a user will click~\cite{bakshy12,tucker12}.

This type of advertisements are usually associated with a \emph{cost per
engagement} (CPE) model. The \emph{advertiser} enters into an
agreement with the platform owner, called the \emph{host}: the advertiser agrees to pay
the host an amount $cpe(i)$ for each click received by  its ad $i$. The clicks
may come not only from the users who saw $i$ as a promoted ad post, but
also their (transitive) followers, who saw it because of viral
propagation.
The agreement also specifies a budget $B_i$, that is, the
advertiser $a_i$ will pay the host the total cost of all the clicks
received by $i$, {\sl up to a maximum of $B_i$}. Naturally, posts from
different advertisers may be promoted by the host concurrently.

Given that promoted posts are inserted in the timeline of the users,
they compete with organic social posts and with one another for a user's
attention. A large number of promoted posts (ads) pushed to a user by
the system would disrupt user experience, leading to disengagement
and eventually abandonment of the platform. To mitigate this, the host
limits the number of promoted posts that it shows to a user within a
fixed time window, e.g., a maximum of 5 ads per day per user: we call
this bound the \emph{user-attention bound}, $\kappa_u$, which may be
user specific~\cite{linHWY14}.

\eat{ 
Assuming an average attention bound of $\kappa$ for all users,
and considering a number $n$ of active users in the social networking
platform, the host has $n\kappa$ daily slots to allocate. 
However, the host might have a much larger demand from advertisers. Suppose that
in a day the host has to serve a set $ \{a_1,\cdots,a_h\}$ of $h$
advertisers, each one having a specific cost-per-engagement and budget.
Then the total number of potential clicks that the host might profit
from is $\sum_{i=1}^h \frac{B_i}{cpe(i)}$, which could be much larger than
$n\kappa \cdot ctr$, where $ctr$ is the average click through rate on ads, which
is known to be low (around 1\% -- 3\%) and can be improved by social proofs~\cite{bakshy12,tucker12}. 
} 

A subtle point here 
is that ads directly promoted by the host count
against user attention bound. On the contrary, an ad $i$ that flows from a
user $u$ to her follower $v$ should not count toward $v$'s attention
bound. In fact,  $v$ is receiving ad $i$ from user $u$, whom she is
voluntarily following: as such, it cannot be considered ``promoted''.

A na\"{\i}ve ad allocation\footnote{In the rest of the paper we use the form ``allocating ads to users'' as well as ``allocating users to ads''  interchangeably.} would match
each ad with the users most likely to click on
the ad. However, the above strategy
fails to leverage the possibility of  ads propagating virally from endorsing
users to their followers. 
We next illustrate the gains achieved by an allocation that takes viral ad propagation into
account. 

\spara{Viral ad propagation: why it matters.} \enlargethispage*{\baselineskip}
For our example we use the toy social network in Fig.~\ref{fig:viral-ad-ex}.
We assume that each time a user clicks on a promoted post, the system produces a social proof  for such engagement action, thanks to which her followers might be influenced to click as well.

In order to model the propagation of (promoted) posts in the network, we can
borrow from the rich body of work done in diffusion of information and
innovations in social networks.
In particular, the
\emph{Independent Cascade} (IC) model~\cite{kempe03}, adapted to our setting, says that once a user $u$
clicks on an ad, she has one independent attempt to try to influence
each of her neighbors $v$. Each attempt succeeds with a probability $p_{u,v}^i$
which depends on the topics of the specific ad $i$ and the influence exerted by $u$ on her neighbor $v$.
The propagation stops when no new users get influenced.
Similarly, we model  the \emph{intrinsic relevance} of a promoted post $i$ to a user $u$,
as the probability $\defaultprob(u,i)$ that $u$ will click on ad
$i$, based on the content of the ad and her own interest profile, i.e.,  the prior probability that the user will click on
a promoted post in the absence of any social proof.

Since the model is probabilistic, we focus on the number of clicks that an ad receives in \emph{expectation}.
Formal details of the propagation model, the topic model, and the definition
of expected revenue are deferred to \textsection~\ref{sec:formalprob}.

\begin{figure}[t!]
\vspace{-2mm}

\begin{framed}
\begin{small}
\begin{center}
\begin{itemize} 
\item Ads = $\{a, b, c, d\}$ 
\item $p_{uv}^i$ (on the edges) are the same  $\forall i \in \{a,b,c,d\}$
\item $ \forall u\in \{v_1,\ldots,v_6\}: \;\; \defaultprob(u,a) = 0.9, \; \; \defaultprob(u,b) = 0.8,$ \\ $\defaultprob(u,c) = 0.7, \; \; \defaultprob(u,d) = 0.6 $
\item $B_a=4, B_b=2, B_c=2, B_d=1$
\item $\kappa_u = 1\;\; \forall u\in \{v_1,\ldots,v_6\}$
  \end{itemize}

\includegraphics[width=.7\textwidth]{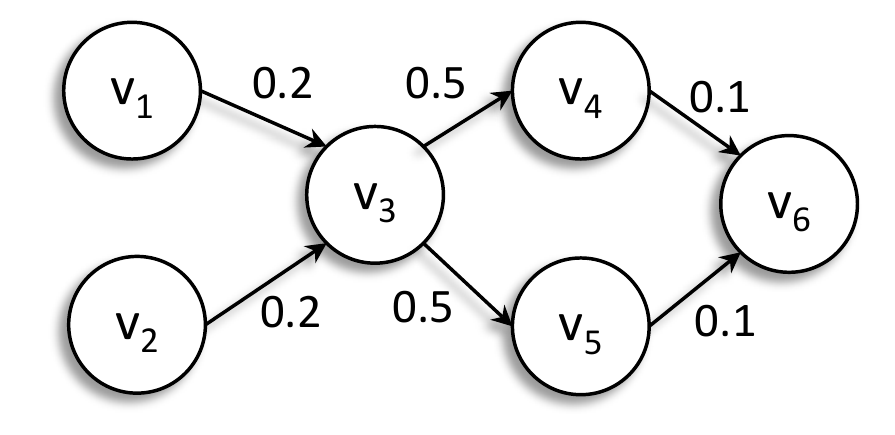}
\end{center}

\textbf{Allocation $\mathcal{A}$}: \emph{maximizing} $\defaultprob(u,i)$\\
$\langle v_1,a\rangle, \langle v_2,a\rangle, \langle v_3,a\rangle, \langle v_4,a\rangle, \langle v_5,a\rangle, \langle v_6,a\rangle$ \\
\hrule
\smallskip

$Pr^{{\cal A}}(click(v_1,a))  = Pr^{{\cal A}}(click(v_2,a)) = 0.9$\\
$Pr^{{\cal A}}(click(v_3,a)) = 1 - (1-0.9\cdot 0.2)^2(1-0.9) = 0.93$\\
$Pr^{{\cal A}}(click(v_4,a)) = Pr^{{\cal A}}(click(v_5,a)) = 1 - (1-0.93\cdot 0.5)(1-0.9) = 0.95$\\
$Pr^{{\cal A}}(click(v_6,a)) = 1 - (1-0.95\cdot 0.1)^2(1-0.9) = 0.92$\\
\vspace{-2mm}

\textbf{Expected number of clicks} = $2\times 0.9 + 0.93 + 2 \times 0.95 +
0.92 = {\bf 5.55}$\\

\hrule
\smallskip
\smallskip

\textbf{Allocation $\mathcal{B}$}: \emph{leveraging virality}\\
$\langle v_1,a\rangle, \langle v_2,a\rangle, \langle v_3,b\rangle, \langle v_4,c\rangle, \langle v_5,c\rangle, \langle v_6,d\rangle$\\
\hrule

\smallskip

$Pr^{{\cal B}}(click(v_1,a))  = Pr^{{\cal B}}(click(v_2,a)) = 0.9$\\
$Pr^{{\cal B}}(click(v_3,a)) = 1 - (1-0.9\cdot 0.2)^2 = 0.33$\\
$Pr^{{\cal B}}(click(v_4,a)) = Pr^{{\cal B}}(click(v_5,a))  = 0.33\cdot 0.5 = 0.16$\\
$Pr^{{\cal B}}(click(v_6,a)) = 1 - (1 - 0.16\cdot 0.1)^2 = 0.03$\\
$Pr^{{\cal B}}(click(v_3,b))  = 0.8$\\
$Pr^{{\cal B}}(click(v_4,b))  = Pr^{{\cal B}}(click(v_5,b)) =  0.8\cdot 0.5 = 0.4$\\
$Pr^{{\cal B}}(click(v_6,b)) = 1 - (1 - 0.8\cdot 0.5\cdot 0.1)^2 = 0.08$ \\
$Pr^{{\cal B}}(click(v_4,c))  = Pr^{{\cal B}}(click(v_5,c)) = 0.7$ \\
$Pr^{{\cal B}}(click(v_6,c))  = 1 - (1-0.7\cdot 0.1)^2 = 0.14$ \\
$Pr^{{\cal B}}(click(v_6,d))  =  0.6$ \\
\vspace{-2mm}

\textbf{Expected number of clicks} = $2\cdot 0.9 + 0.33 + 2\cdot 0.16 + 0.03 + 0.8 + 2\cdot 0.4 + 0.08 + 2\cdot 0.7 + 0.14 + 0.6 = {\bf 6.3}$.
\end{small}
\end{framed}
\vspace{-4mm}

  \caption{Illustrating viral ad propagation.  For simplicity, we round all numbers to the second decimal. \label{fig:viral-ad-ex}}

\vspace{-3mm}
\end{figure}

Consider the example in Fig.~\ref{fig:viral-ad-ex}, where we assume peer influence probabilities (on edges) are equal for all the four ads $\{a,b,c,d\}$.
The figure also reports $\defaultprob(u,i)$ and advertiser budgets. For each advertiser, CPE is 1  and the attention bound for every user is $1$, i.e., no user wants
more than one ad promoted to her by the host.
The expected revenue for an allocation is the same as the resulting expected number of clicks, as the CPE is $1$.
Below, for simplicity, we round all numbers to the second decimal \emph{after} calculating them all.

Let us consider two ways of allocating users to ads by the host. In
allocation $\cala$, the host matches each user to her top preference(s) based
on $\defaultprob(u,i)$, subject to not violating the attention bound.
This results in ad $a$ being assigned to all six users, since it has the
highest engagement probability for every user. No further ads may be
promoted without violating the attention bound. In allocation $\calb$, the
host recognizes viral propagation of ads and thus assigns $a$ to $v_1$
and $v_2$, $b$ to $v_3$, $c$ to $v_4$ and $v_5$, and $d$ to $v_6$.

Under allocation $\cala$, clicks on $a$ may come from all six users: $v_1, v_2$
click on $a$ with probability $0.9$. However, $v_3$ clicks on $a$ w.p. $(1 -
(1-0.9\cdot 0.2)^2(1-0.9)) = 0.93$.
This is obtained by combining three factors: $v_3$'s engagement probability of $0.9$ with ad $a$, and probability $0.9\cdot 0.2$ with which each of $v_1, v_2$ clicks on $a$ and influences $v_3$ to click on $a$.  In a similar
way one can derive the probability of clicking on $a$ for $v_4$, $v_5$, and $v_6$ (reported in the figure).
The overall expected revenue for allocation $\cala$ is the sum of all clicking probabilities: $2\times 0.9 + 0.93 + 2 \times 0.95 +
0.92 = {\bf 5.55}$.

Under allocation $\calb$, the ad $a$ is promoted to only $v_1$ and $v_2$ (which click on it w.p. $0.9$). Every other user that clicks on $a$ does so solely based on social influence. Thus, $v_3$ clicks on $a$ w.p. $1 - (1-0.9\cdot 0.2)^2 = 0.33$. Similarly one can derive the probability of clicking on $a$ for $v_4$, $v_5$, and $v_6$ (reported in the figure).
Contributions to the clicks on $b$ can only come from nodes $v_3, v_4, v_5, v_6$. They click on $b$, respectively, w.p. $0.8$, $0.8\cdot 0.5 = 0.4$, $0.8\cdot 0.5 = 0.4$, and $1 - (1 -
0.8\cdot 0.5\cdot 0.1)^2 = 0.08$.

Finally, it can be verified that the expected number of clicks on ad $c$ is $0.7 + 0.7 + (1 - (1-0.7\cdot 0.1)^2)$, while on  $d$ is just $0.6$. The overall number of expected clicks under allocation $\calb$ is  {\bf 6.3}.

\smallskip
\noindent
{\bf Observations}: (1) Careful allocation of users to ads that takes viral ad propagation into account can outperform an allocation that merely focuses on immediate clicking likelihood based on the content relevance of the ad to a user's interest profile. It is easy to construct instances where the gap between the two can be arbitrarily high by just replicating the gadget in Fig.~\ref{fig:viral-ad-ex}. 

(2) Even though allocation $\cala$ ignores the effect of viral ad propagation, it still benefits from the latter, as shown in the calculations. This naturally motivates finding allocations that expressly exploit such propagation in order to maximize the expected revenue.

\eat{
(3) The benefit of allocation $\calb$ over $\cala$ carries over to any attention bounds $\kappa_u$, as long as $\kappa_u \ll$ the number of ads available, which is typically the case in a real-world scenario.
}


In this context, we study the problem of {\sl how to strategically allocate users to the advertisers, leveraging social influence and the propensity of ads to propagate}.
The major challenges in solving this problem are as follows. Firstly, the host needs to strike a balance between assigning ads to users who are likely to click and assigning them to ``influential'' users who are likely to boost further propagation of the ads. Moreover, influence may well depend on the ``topic'' of the ad. E.g., $u$ may influence its neighbor $v$ to different extents on cameras versus health-related products. Therefore, ads which are close in a topic space \emph{will naturally compete} for 
users that are influential in the same area of the topic space. Summarizing, a good allocation strategy needs to take into account the different CPEs and budgets for different advertisers, users'  attention bound and interests, and ads' topical distributions.

 An even more complex challenge is brought in by the fact that uncontrolled virality could be undesirable for the host, as it creates room for exploitation by the advertisers: hoping to tap uncontrolled virality, an advertiser might declare a lower budget for its marketing campaign, aiming at the same large outcome with a smaller cost. Thus, from the host perspective, it is important to make sure the expected revenue from an advertiser is as close to the budget as possible:
both undershooting and overshooting the budget results in a \emph{regret} for the host, as illustrated in the following example.

\vspace*{-1ex}
\begin{example}\label{ex2}
Consider again our example in Fig.~\ref{fig:viral-ad-ex}, but this time along with the budgets specified the four advertisers are $B_a = 4, B_b = 2, B_c = 2, B_d = 1$. Then, rounding to the first decimal, allocation $\cala$ leads to an  overall regret of 
$|4-5.6| + |2-0| + |2-0| + |1-0| = {\bf 6.6}$: the expected revenue exceeds the budget for advertiser $a$ by $1.6$ and falls short of other advertiser budgets by $2, 2, 1$ respectively. Similarly, for allocation $\calb$, the regret is 
$|4-2.5| + |2-1.7| + |2-1.5| + |1-0.6| = {\bf 2.7}$.
\qed
\end{example}

\vspace*{-1ex}
The host knows it will not be paid beyond the budget of each advertiser, so that any excess above the budget is essentially ``free service'' given away by the host, which causes regret, and any shortfall w.r.t. the budget is a lost revenue opportunity which causes regret as well.
This creates a challenging trade-off: on the one hand, the host aims at leveraging virality and the network effect to improve advertising efficacy, while on the other hand the host wants to avoid giving away free service due to uncontrolled virality.

\spara{Contributions and roadmap.}
In this paper we make the following major contributions:
\squishlist
\item We propose a novel problem domain of allocating users to advertisers for promoting advertisement posts, taking advantage of network effect, while paying attention to important practical factors like relevance of ad, effect of social proof, user's attention bound, and limited advertiser budgets  (\textsection~\ref{sec:formalprob}).
\item We formally define the problem of \emph{minimizing regret} in allocating users to ads (\textsection~\ref{sec:formalprob}), and  show that it is \NPhard and is \NPhard to approximate within any factor (\textsection~\ref{sec:theory}).
\item We develop a simple greedy algorithm and establish an upper bound on the regret it achieves as a function of advertisers' total budget (\textsection~\ref{sec:greedy}).
\item We then devise a scalable instantiation of the greedy algorithm by leveraging the notion of \emph{random reverse-reachable sets}~\cite{borgs14,tang14} (\textsection~\ref{sec:algo}).
\item Our extensive experimentation on  four real datasets confirms that our algorithm is scalable and it delivers high quality solutions,  significantly outperforming  natural baselines (\textsection~\ref{sec:exp}).
\squishend

{\sl To the best of our knowledge, regret minimization in the context of promoting multiple ads in a social network, subject to budget and attention bounds has not been studied before.}
Related work is discussed in \textsection~\ref{sec:related},
\eat{while in \textsection~\ref{sec:formalprob} we introduce the needed concepts and we formalize the problem studied.  \textsection~\ref{sec:theory} presents our theoretical results at the basis of our algorithm, while \textsection~\ref{sec:algo} focuses on how to make our results scalable to real-world settings.
\textsection~\ref{sec:exp} contains the experimental assessment of our proposal,}
while \textsection~\ref{sec:concl} concludes the paper discussing extensions and future work.

\section{Related Work}
\label{sec:related}
\enlargethispage*{\baselineskip}
Substantial work has been done on viral marketing, which mainly focuses
on a key algorithmic problem -- \emph{influence maximization}~\cite{
kempe03, ChenWW10, goyal12}.
Kempe et al. \cite{kempe03} formulated influence maximization as
a discrete optimization problem: given a social graph and a number $k$, find a set $S$ of $k$ nodes, such that by activating them one maximizes the expected
spread of influence $\sigma(S)$ under a certain propagation model, e.g., the {\em Independent Cascade} (IC)
 model. Influence maximization is \NPhard, but the function $\sigma(S)$ is
\emph{monotone}\footnote{$\sigma(S) \leq \sigma(T)$ whenever $S \subseteq T$.}  and \emph{submodular}\footnote{$\sigma(S \cup \{w\}) - \sigma(S) \geq \sigma(T \cup \{w\}) -
\sigma(T)$  whenever $S \subseteq T$.}~\cite{kempe03}.
Exploiting these properties, the simple greedy algorithm that at each step extends the seed set with  the node providing the largest marginal gain, provides a $(1 - 1/e)$-approximation to the optimum \cite{submodular}.
The greedy algorithm is computationally prohibitive, since selecting the node with the largest marginal gain is \SPhard~\cite{ChenWW10}, and is typically approximated by numerous Monte Carlo simulations~\cite{kempe03}.
However, running many such simulations is extremely costly, and thus considerable effort has been devoted to developing efficient and scalable influence maximization algorithms: in \textsection\ref{sec:algo} we will review some of the latest advances in this area which help us devise our algorithms.

Datta et al.~\cite{datta2010viral} study influence maximization with multiple items, under a user attention constraint. However, as in classical influence maximization, their objective is to maximize the overall influence spread, and the budget is w.r.t. the size of the seed set, so without any CPE model. Their diffusion model is the (topic-blind) IC model, which also doesn't model the competition among similar items. They propose a simple greedy approximation algorithm and a heuristic algorithm for fair allocation of seeds with no guarantees.
Du et al.~\cite{du2014multiple} study influence maximization over multiple non-competing products subject to user attention constraints and product budget (knapsack) constraints, and develop approximation algorithms in a continuous time setting. 
A noteworthy feature of our work is that, as will be shown in \textsection \ref{sec:exp}, the budgets we use are such that thousands of seeds are required to minimize regret. Scalability of algorithms for selecting thousands of seeds over large networks has not been demonstrated before. Lin et al.~\cite{linHWY14} study the problem of maximizing influence spread from a website's perspective: how to dynamically push items to users based on user preference and social influence.
The push mechanism is also subject to user attention bounds.
Their framework is based on Markov Decision Processes (MDPs).

Our work departs from the body of work in this field by looking at the possibility of
integrating viral marketing into existing social advertising models and by
studying a fundamentally different objective: \emph{minimize host's regret}.

While social advertising is still in its infancy, it fits in the more general (and mature) area of computational advertising that has attracted a lot of interest during the last decade. The central problem of computational advertising is to find the ``best match'' between a given user in a given context and a suitable advertisement. The context could be a user entering a query in a search engine (``sponsored search"), reading a web page (``content match" and ``display ads"), or watching a movie on a portable device, etc.

The most typical example is sponsored search:
search engines show ads deemed relevant to user-issued queries,
in the hope of maximizing click-through rates and in turn, revenue.
Revenue maximization in this context is formalized as
the well-known {\em Adwords} problem~\cite{adwords}.
We are given a set $Q$ of keywords and $N$ bidders with their daily budgets and bids for each keyword in $Q$.
During a day, a sequence of words (all from $Q$) would arrive online and the task is to assign each word to one bidder {\sl upon its arrival}, with the objective of maximizing revenue for the given day while respecting the budgets of all bidders.
This can be seen as a generalized online bipartite matching problem, and by using linear programming techniques, a $(1-1/e)$ competitive ratio is achieved~\cite{adwords}.
Considerable work has been done in sponsored search and display ads~\cite{goel08, feldman09,feldman10,mirrokni12,devanur12}.
For a comprehensive treatment, see a recent survey~\cite{mehta13}.
Our work fundamentally differs from this as we are concerned with the {\em virality} of ads when making allocations: this concept is still largely unexplored in computational advertising.
\enlargethispage*{\baselineskip}

Recently,
Tucker \cite{tucker12} and Bakshy et al.~\cite{bakshy12} conducted field experiments on Facebook and demonstrated that adding social proofs to sponsored posts in Facebook's News Feed significantly increased the click-through rate.
Their findings empirically confirm the benefits of social influence, paving the way for the application of viral marketing in social advertising, as we do in our work.

\section{Problem statement}
\label{sec:formalprob}
\smallskip\noindent
{\bf The Ingredients.}
The computational problem studied in this paper is from the host perspective. The host 
owns: (i) a \emph{directed social graph} $G=(V,E)$, where an arc $(u,v)$ means that $v$ follows $u$, thus $v$ can see $u$'s posts and can be influenced by $u$; 
(ii) a \emph{topic model} for ads and users' interest, defined on  a space of $K$ topics; 
(iii) a \emph{topic-aware influence propagation model} defined on the social graph $G$ and the topic model.

The key idea behind the topic modeling
is to introduce a hidden variable $Z$ that can
range among $K$ states. Each topic (i.e., state of the latent
variable) represents an abstract interest/pattern and intuitively
models the underlying cause for each data observation (a user clicking on an ad).
In our setting the host owns a precomputed probabilistic topic model. The actual method used for producing the model is not important at this stage: it could be, e.g.,
 the popular \textit{Latent Dirichlet Allocation} (LDA)~\cite{blei:lda}, or any other method. What is relevant is that the topic model  maps each ad $i$ to a topic
distribution $\vec{\gamma_i}$ over the latent topic space, formally:
$\gamma_i^z = Pr(Z=z|i) \mbox{ with } \Sigma_{z = 1}^K\gamma_i^z = 1.$

\smallskip\noindent
{\bf Propagation Model.}
The propagation model governs the way that ads propagate in the social network driven by social influence. In this work, we extend a simple topic-aware propagation model introduced by Barbieri et al.~\cite{BarbieriBM12}, with Click-Through Probabilities (CTPs) for seeds: we refer to the set of users $S_i$ that receive ad $i$ directly as a promoted post from the host as the \emph{seed set} for ad $i$.
In the \emph{Topic-aware Independent Cascade} model (TIC) of~\cite{BarbieriBM12}, the propagation proceeds as follows: when a node $u$ first clicks an ad $i$, it has one chance of influencing each inactive neighbor $v$, independently of the history thus far. This succeeds with a probability that is the weighted average of the arc probability w.r.t. the topic distribution of the ad $i$:
\begin{equation}\label{eq:tic}
	p^i_{u,v} = \sum\nolimits_{z = 1}^K\gamma_i^z \cdot p_{u,v}^z.
\end{equation}

For each topic $z$ and for a seed node $u$, the probability $p_{H,u}^z$  represents the likelihood of $u$ clicking on a promoted post for topic $z$.
Thus the CTP $\defaultprob(u,i)$ that $u$ clicks on the promoted post $i$ in absence of any social proof,  is the weighted average (as in Eq.~\eqref{eq:tic}) of the probabilities $p_{H,u}^z$ w.r.t. the topic distribution of $i$.
In our extended TIC-CTP model, each $u \in S_i$ accepts to be a seed, i.e., clicks on ad $i$, with probability $\delta(u,i)$ when targeted.
The rest of the propagation process remains the same as in TIC.

Following the literature on influence maximization we denote with  $\sigma_i(S_i)$ the \emph{expected number of clicks} (according to the TIC-CTP model) for ad $i$ when the seed set is $S_i$. The corresponding expected revenue is
$\Pi_i(S_i)=\sigma_i(S_i)\cdot cpe(i)$, where $cpe(i)$ is the
cost-per-engagement that $a_i$ and the host have agreed on.

We observe that for a fixed ad $i$, with topic distribution  $\vec{\gamma_i}$, the TIC-CTP model boils down to the standard \emph{Independent Cascade} (IC) model \cite{kempe03} with CTPs, where again, a seed may activate with a probability.
We next expose the relationship between the expected spread a $\sigic(S)$ for the classical IC model without CTPs, and the expected spread  under the TIC-CTP model for a given ad $i$.
\def\lemmaICCTP{
Given an instance of the TIC-CTP model, and a fixed ad $i$, with topic distribution  $\vec{\gamma_i}$, build an instance of IC by setting the probability over each edge $(u,v)$ as in Eq.~\ref{eq:tic}. Now,  consider any node $u$,  and any set $S$ of nodes.
Let $\delta(u,i)$ be the CTP for $u$ clicking on the promoted post $i$.
Then we have
\begin{align}
\delta(u,i) [ \sigic(S \cup \{u\}) - \sigic(S) ]
= \sigma_i(S \cup \{u\}) - \sigma_i(S).
\end{align}
}

\begin{lemma}\label{lem:icctp}
{\lemmaICCTP}
\end{lemma}

\begin{proof}
The proof relies on the possible world semantics.
For the IC model~\cite{kempe03}, consider a graph $G=(V,E)$ with influence probability $p_{u,v}$ on each edge $(u,v) \in E$.
A possible world, denoted $X$, is a deterministic graph generated as follows.
For each edge $(u,v) \in E$, we flip a biased coin: with probability $p_{u,v}$, the edge is declared ``live'', and with probability $1-p_{u,v}$, it is declared ``blocked''.

Define an indicator function $\mathbb{I}_X(S,v)$, which takes on $1$ if $v$ is reachable by $S$ via a path consisting entirely of live edges in $X$, and $0$ otherwise.
In the IC model,
\begin{align*}
&\qquad \sigic(S \cup \{u\}) - \sigic(S)  \\
&= \sum_X \Pr[X] \cdot ( |\{w : \mathbb{I}_X(S \cup \{u\}, w) = 1\}| -|\{w : \mathbb{I}_X(S, w) = 1\}| ) \\
&= \sum_X \Pr[X] \cdot |\{w: \mathbb{I}_X(S \cup \{u\}, w) = 1 \, \wedge \, \mathbb{I}_X(S, w) = 0 \}| \\
&= \sum_X \Pr[X] \cdot |\{w: \mathbb{I}_X(\{u\}, w) = 1 \, \wedge \, \mathbb{I}_X(S, w) = 0 \}|.
\end{align*}

Notice that for a node to be active in a possible world, it must be reachable from a seed. In each of the possible worlds, node $u$ has probability $\delta(u,i)$ to accept to become a seed. 
Thus, in the TIC-CTP model, we have:
\begin{align*}
&\qquad \sigma_i(S \cup \{u\}) - \sigma_i(S) \\
&= \delta(u, i) \cdot \sum_X \Pr[X] \cdot |\{w: \mathbb{I}_X(\{u\}, w) = 1 \, \wedge \, \mathbb{I}_X(S, w) = 0 \}|.
\end{align*}

This directly leads to
$$
\delta(u, i) (\sigic(S \cup \{u\}) - \sigic(S)) = \sigma_i(S \cup \{u\}) - \sigma_i(S),
$$
which was to be shown.
\end{proof}

A corollary of the above lemma is that for a fixed $\vec{\gamma_i}$, the expected spread  $\sigma_i(\cdot)$ function under the TIC-CTP model, inherits the properties of monotonicity and submodularity from the IC model (see Sec.~\ref{sec:related} and \cite{kempe03,BarbieriBM12}). In turn, $\Pi_i(S_i) = cpe(i) \cdot \sigma_i(S_i)$ is also monotone and submodular, being a non-negative linear combination of monotone submodular functions.


\smallskip\noindent
{\bf Budget and Regret.} As in any other advertisement model, we assume that
each advertiser $a_i$ has a finite budget $B_i$ for a campaign on ad $i$, which limits the maximum amount  that $a_i$ will
pay the host. The host needs to allocate seeds to each of the ads that
it has agreed to promote, resulting in an allocation $\cals = (S_1, ...,
S_h)$. 
The expected revenue from the campaign may fall short of the budget (i.e.,
$\Pi_i(S_i) < B_i$) or overshoot it (i.e., $\Pi_i(S_i) > B_i$). An
advertiser's natural goal is to make its expected revenue as close to
$B_i$ as possible: the former situation is lost opportunity to make
money whereas the latter amounts to ``free service'' by the host to the
advertiser. Both are undesirable. Thus, one option to define the host's
regret for seed set allocation $S_i$ for advertiser $a_i$ is as $\left|B_i
- \Pi_i(S_i) \right|$.

Note that this definition of regret has the drawback that it does not discriminate between small and large seed sets: given two
seed sets $S_1$ and $S_2$ with the same regret as defined above, and with $|S_1| \ll |S_2|$,
this definition does not prefer one over the other. In practice, it is desirable to achieve a low regret with a small number of seeds. By drawing on the inspiration from the optimization literature~\cite{Boyd:2004:CO}, where an additional penalty corresponding to the complexity of the solution is added to the
error function to discourage overfitting, we propose to add a similar penalty term to discourage the use of large seed sets. Hence we define the \emph{overall regret} as
\begin{equation}\label{eq:reg}
	\calr_i(S_i) = \left|B_i - \Pi_i(S_i) \right| + \lambda \cdot |S_i|.
\end{equation}

Here, $\lambda\cdot|S_i|$ can be seen as a penalty for the use of a seed set: the larger its size, the greater the penalty. This discourages the choice of a large number of poor quality seeds to exhaust the budget. When $\lambda=0$, no penalty is levied and the ``raw'' regret corresponding to the budget alone is measured. We assume w.l.o.g. that the scalar $\lambda$ encapsulates CPE such that the term $\lambda|S_i|$ is in the same monetary unit as $B_i$. How small/large should $\lambda$ be? We will address this question in the next section.

The overall regret from an allocation $\cals = (S_1, ..., S_h)$ to all
advertisers is
\begin{align}\label{eq:total_regret}
	\calr(\cals) = \sum\nolimits_{i=1}^h \calr_i(S_i).
\end{align}

\begin{example}\label{ex3}
In Example~\ref{ex2}, the regrets reported for allocations $\cala$ ($6.6$) and $\calb$ ($2.7$) correspond to $\lambda=0$. When $\lambda=0.1$, the regrets change to $6.6 + 0.1\times 6 = 7.2$ for $\cala$ and to $2.7+0.1\times 6 = 3.3$ for $\calb$. \qed
\end{example}

%

As noted in the introduction, in practice, the number of ads that can be
promoted to a user may be limited. The host can even personalize this
number depending on users' activity.
We model this using an attention bound $\kappa_u$ for user $u$.
An allocation $\cals = (S_1, ...,
S_h)$ is called \emph{valid} provided for every user $u\in V$, $|\{S_i \in
\cals \mid u\in S_i\}| \le \kappa_u$, i.e., no more than $\kappa_u$ ads
are promoted to $u$ by the allocation. We are now ready to formally state
the problem we study.

\begin{problem}[\MRA]\label{pr:noSeedCosts}
We are given $h$ advertisers $a_1, \ldots, a_h$, where each $a_i$
has an ad $i$ described by topic-distribution $\vec{\gamma_i}$, a
budget $B_i$, and a cost-per-engagement $cpe(i)$.
Also given is a social graph $G=(V,E)$ with a
probability $p_{u,v}^z$ for each edge $(u,v) \in E$ and each topic $z
\in [1,K]$, an attention bound $\kappa_u$, $\forall u\in V$, and a penalty
parameter $\lambda \ge 0$. The task is to compute a valid
allocation $\cals = (S_1, \ldots, S_h)$ that minimizes the
overall regret:
$$
\cals = \argmin_{\substack{\mathcal{T}=(T_1,\ldots,T_h):T_i \subseteq V \\ \mathcal{T} \mbox{ is valid}}} \calr(\mathcal{T}).
$$
\end{problem}


\noindent
\textbf{Discussion.}
Note that $\Pi_i(S_i)$ denotes the expected revenue from advertiser
$a_i$. In reality, the actual revenue depends on the number of
engagements the ad \emph{actually} receives. Thus, the uncertainty in
$\Pi_i(S_i)$ may result in a loss of revenue.
Another concern could be that regret on the positive side ($\Pi_i(S_i) >
B_i$) is more acceptable than on the negative side ($\Pi_i(S_i) <
B_i$), as one can argue that maximizing revenue is a more critical goal
even if it comes at the expense of a small and reasonable amount of free
service.
Our framework can accommodate such concerns and can easily address them.
For instance, instead of defining raw regret as $|B_i -
\Pi_i(S_i)|$, we can define it as $|B_i' - \Pi_i(S_i)|$, where $B_i' =
(1+\beta) \cdot B_i$. The idea is to artificially boost the budget $B_i$
with parameter $\beta$ allowing maximization of revenue while keeping
the free service within a modest limit. This small change has no
impact on the validity of our results and algorithms.
\eat{For instance,
Theorem~\ref{thm:greedyOneThird} shows that our algorithm achieves a regret of at most $\sum_i B_i/3$. If we decide
to use boosted budgets $B_i'$ instead of $B_i$, the upper bound on regret
would then be $\sum_i B_i'/3$. The other theorems can be extended in a similar fashion.
}
Theorem~\ref{thm:full-regret} provides an upper bound on the regret achieved by our allocation algorithm (\textsection~\ref{sec:greedy}). The bound remains intact except that in place of the original budget $B_i$, we should use the boosted budget $B_i'$. This remark applies to all our results.
We henceforth study the problem as defined in Problem
\ref{pr:noSeedCosts}.

\section{Theoretical Analysis}
\label{sec:theory}
\enlargethispage*{\baselineskip}
We first show that  \MRA\ is
not only \NPhard to solve optimally, but is also \NPhard to approximate
within any factor (Theorem~\ref{thm:mra-hard}).
On the positive side, we propose a greedy algorithm and conduct a careful
analysis to establish a bound on the regret it can achieve as a function of the budget (Theorems~\ref{thm:full-regret}-\ref{thm:greedyBetter}).

\eat{In particular, we show that our greedy algorithm guarantees a
regret of no more than $\frac{1}{3}\sum_{i} B_i$
(Theorem~\ref{thm:greedyOneThird}), i.e., within a third of the total
budget. We also show that our algorithm in fact provides better
guarantees, depending on input instances
(Theorem~\ref{thm:greedyBetter}). Click-through-probabilities are orthogonal to the results of this section and to the greedy algorithm presented here. Therefore, in this section, for simplicity, we will assume all CTPs are $1$ without loss of generality. They are addressed carefully in Section~\ref{sec:algo}.
We start with the following hardness result. }

\begin{theorem}\label{thm:mra-hard}\label{thm:mra-inapprox}
\MRA is \NPhard and is \NPhard to approximate within any factor.
\end{theorem}

\begin{proof}
We prove hardness for the special case where $\lambda = 0$, using a
reduction from \textsc{3-Partition}~\cite{Garey1979}.

Given a set $X=\{x_1, ..., x_{3m}\}$ of positive integers whose sum is
$C$,  with $x_i \in (C/4m, C/2m)$, $\forall i$, \textsc{3-Partition}
asks whether $X$ can be partitioned into $m$ disjoint 3-element subsets,
such that the sum of elements in each partition is the same (= $C/m$).
This problem is known to be strongly \NPhard, i.e., it remains
\NPhard even if the integers $x_i$ are bounded above by a polynomial in
$m$~\cite{Garey1979}. Thus, we may assume that $C$ is bounded by a polynomial
in $m$.

Given an instance \cali\ of \textsc{3-Partition}, we construct an instance \calj\ of \MRA
as follows. First, we set the number of advertisers $h=m$ and let the
cost-per-engagement (CPE) be $1$ for all advertisers. Then, we construct a
directed bipartite graph $G=(U \cup V, E)$: for each number $x_i$, $G$
has one node $u_i\in U$ with $x_i-1$ outneighbors in $V$, with all
influence probabilities set to $1$. We refer to members of $U$ (resp.,
$V$) as ``$U$'' nodes (resp., ``$V$'' nodes) below. Set all advertiser
budgets to $B_i = C/m$, $1\le i\le m$ and the attention bound of every
user to $1$. This will result in a total of $C$ nodes in the instance of
\MRA. Since $C$ is bounded by a polynomial in $m$, the reduction is
achieved in polynomial time.

We next show that if \MRA can be solved in polynomial time, so can
\textsc{3-Partition}, implying hardness.  To that end, assume there
exists an algorithm ${\bf A}$ that solves \MRA optimally. We can use ${\bf A}$
to distinguish between YES- and NO-instances of \textsc{3-Partition} as
follows. Run ${\bf A}$ on \calj\ to yield a seed set allocation $\cals =
(S_1, ..., S_m)$. We claim that \cali\ is a YES-instance of
\textsc{3-Partition} iff $\calr(\cals) = 0$, i.e., the total regret of
the allocation \cals\ is zero.

\noindent
($\Longrightarrow$): Suppose $\calr(\cals) = 0$. This implies the regret of
every advertiser must be zero, i.e., $\Pi_i(S_i) = B_i = C/m$. We shall
show that in this case, each $S_i$ must consist of 3 ``$U$'' nodes whose
spread sums to $C/m$. From this, it follows that the 3-element subsets
$X_i := \{x_j \in X \mid u_j \in S_i\}$ witness the fact that \I\ is a
YES-instance. Suppose $|S_i| \ne 3$ for some $i$. It is trivial to see
that each seed set $S_i$ can contain only the ``$U$'' nodes, for the
spread of any ``$V$'' node is just $1$. If $|S_i| \ne 3$, then
$\Pi_i(S_i) = \sum_{u_j\in S_i} x_j \ne C/m$, since all numbers are in
the open interval $(C/4m, C/2m)$. This shows that every seed set $S_i$
in the above allocation must have size $3$, which was to be shown.

\noindent
($\Longleftarrow$): Suppose $X_1, ..., X_m$ are disjoint 3-element subsets
of $X$ that each sum to $C/m$. By choosing the corresponding
``$U$''-nodes we get a seed set allocation whose total regret is zero.

We just proved that \MRA is \NPhard. To see hardness of approximation,
suppose ${\bf B}$ is an algorithm that approximates \MRA within a factor of
$\alpha$. That is, the regret achieved by algorithm ${\bf B}$ on any
instance of \MRA is $\le \alpha\cdot OPT$, where $OPT$ is the optimal
(least) regret. Using the same reduction as above, we can see that the
optimal regret on the reduced instance \calj\ above is $0$. On this
instance, the regret achieved by algorithm ${\bf B}$ is $\le \alpha\cdot 0
= 0$, i.e., algorithm ${\bf B}$ can solve \MRA optimally in polynomial
time, which is shown above to be impossible unless $P = NP$.
\end{proof}

%

\subsection{A Greedy Algorithm}
\label{sec:greedy} 

Due to the hardness of approximation of Problem 1, no polynomial
algorithm can provide any theoretical guarantees w.r.t.\ optimal
overall regret.
Still, instead of jumping to heuristics without any guarantee,
we present an intuitive greedy algorithm (pseudo-code in
Algorithm~\ref{alg:rmVanilla}) with theoretical guarantees in terms of the total budget.
It is worth noting that
analyzing regret w.r.t.\ the total budget has real-world relevance, as budget is a concrete monetary and known quantity (unlike optimal value of regret)
which makes it easy to understand regret from a business perspective.

\IncMargin{1em}
\begin{algorithm}[t!]
\caption{Greedy Algorithm}
\label{alg:rmVanilla}
\Indm
{\small
\SetKwInOut{Input}{Input}
\SetKwInOut{Output}{Output}
\SetKwComment{tcp}{//}{}
\Input{$G=(V,E)$; $\lambda$; attention bounds $\kappa_u, \forall u\in V$; items $\vec{\gamma}_i$ with $cpe(i)$
\& budget $B_i$, $i = 1,\ldots,h$; $\delta(u,i), \forall u \forall i$}
\Output{$S_1, \ldots, S_h$}
}
\Indp
{
$S_i \gets \emptyset$, $\forall i = 1,\ldots,h$\ \label{line:rm1} \\
\While{true} {
  {\small
	  $(u, a_i) \leftarrow \argmax_{v,a_j} \mathcal{R}_j(S_j) -
	  \mathcal{R}_j(S_j \cup \{v\}) $, \hspace*{12ex} subject to: \ $|\{S_{\ell} | v \in S_{\ell}\}| < \kappa_v$  \textbf{and}
	  \hspace*{20ex} $\mathcal{R}_j(S_j \cup \{v\}) \le
                                \mathcal{R}_j(S_j))$ \label{line:rm3}
  }

\textbf{if} $(u, a_i)$ is \textbf{null then} return
\textbf{else} $S_i \leftarrow S_i \cup \{u\}$
}
}
\end{algorithm}
\DecMargin{1em}

\eat{
We have shown in the previous section that, while the objective of the \MRA problem is quite different than the influence maximization problem, the greedy approach to both problems share the same operational principle of selecting the nodes with the maximum marginal gain at each iteration of the greedy algorithm. This follows from the fact that per-advertiser regret $\mathcal{R}_i(S_i)$, being a non-monotonic function of absolute differences, shows monotonically decreasing submodular behavior while $\Pi_i(S_i) < B_i$. Thus, identification of the node $v$ that results in the maximum reduction in $\mathcal{R}_i(S_i)$, for each advertiser $i$ in an iteration of the greedy algorithm, is the same as selecting the node $u$ that provides the maximum marginal gain in spread on $G^i=(V,E,p^i)$, hence the maximum marginal gain in the revenue $\Pi_i(S_i)$.
}


The algorithm starts by initializing all seed sets to be empty
(line~\ref{line:rm1}). It keeps selecting and allocating seeds until
regret can no longer be minimized.  In each iteration, it finds a
user-advertiser pair $(u, a_i)$ such that $u$'s attention bound is not
reached (that is, ~$|\{S_i | u \in S_i\}| < \kappa_u$) and adding $u$ to $S_i$
(the seed set of $a_i$) yields the largest decrease in regret among all
valid pairs. Clearly, we want to ensure that regret does not increase in
an iteration (that is, $\mathcal{R}_i(S_i \cup \{u\}) <
\mathcal{R}_i(S_i)$) (line 3). The user $u$ is then added to $S_i$. If
no such pair can be found, that is, regret cannot be reduced further,
the algorithm terminates (line 4).


Before stating our results on bounding the overall regret achieved by the greedy algorithm, we identify extreme (and
unrealistic) situations where no such guarantees may be possible.

\smallskip\noindent
{\bf Practical considerations.}
Consider a network with $n$ users, one advertiser with a CPE of $1$ and
a budget $B \gg n$.
Assume CTPs are all $1$.
Clearly, even if all $n$ users are allocated
to the advertiser, the regret approaches $100\%$ of $B$, as most of the
budget cannot be tapped. At another extreme, consider a dense network with $n$
users (e.g., clique), one advertiser with a cpe of $1$ and a budget $B \ll n$.
Suppose the network has high influence probabilities, with the result
that assigning \emph{any} one seed $u$ to the advertiser will result in
an expected revenue $\Pi(\{u\}) \gg B$. In this case, the allocation with
the least regret is the empty allocation (!) and the regret is exactly
$B$! In many practical settings, the budgets are large enough that the
marginal gain of any one node is a small fraction of the budget and
small enough compared to the network size, in that there are enough
nodes in the network to allocate to each advertiser in order to exhaust
or exceed the budget.


%
%

\subsection{The General Case}

In this subsection, we establish an upper bound on the regret achieved by Algorithm~\ref{alg:rmVanilla}, when every candidate seed has essentially an unlimited attention bound.
For convenience, we refer to the first term in the definition of regret (cf. Eq. \ref{eq:reg}) as \emph{budget-regret} and the second term as \emph{seed-regret}. The first one reflects the regret arising from undershooting or overshooting the budget and the second arises from utilizing seeds which are the host's resources. For a seed set $S_i$ for ad $i$, the \emph{marginal gain} of a node $x\in V\setminus S_i$ is defined as $MG_i(x|S_i) := \Pi_i(S_i\cup\{x\}) - \Pi_i(S_i)$.
By submodularity, the marginal gain of any node is the greatest w.r.t.\ the empty seed set, i.e., $MG_i(x|\emptyset) = \Pi_i(\{x\})$.
Let $p_i$ be the maximum marginal gain of any node w.r.t. ad $i$, as a fraction of its budget $B_i$, i.e., $p_i := \mbox{max}_{x\in V} \, \Pi_i(\{x\})/B_i$. As discussed at the end of the previous subsection, we assume that the network and the budgets are such that $p_i \in (0,1)$, for all ads $i$. In practice, $p_i$ tends to be a small fraction of the budget $B_i$.
Finally, we define $p_{max} := \max_{i=1}^h p_i$ to be the maximum $p_i$ among all advertisers.

\begin{theorem}\label{thm:full-regret}
Suppose that for every node $u$, the attention bound  $\kappa_u \ge h$, the number of advertisers, and that $\lambda \le \delta(u,i)\cdot cpe(i)$, $\forall$ user $u$ and ad $i$.
Then the regret incurred by Algorithm~\ref{alg:rmVanilla} upon termination is at most
\[
\sum_{i=1}^h \frac{p_i B_i  + \lambda}{2} \,+\, \lambda \cdot \sum_{i=1}^h \left( 1 + s^i_{opt} \lceil\ln\frac{1}{p_i/2 - \lambda/2B_i}\rceil \right),
\]
where $s^i_{opt}$ is the smallest number of seeds required for reaching or exceeding the budget $B_i$ for ad $i$.
\end{theorem}

\smallskip\noindent
{\bf Discussion}: The term $\delta(u,i) \cdot cpe(i)$ corresponds to the expected revenue from user $u$ clicking on $i$ (without considering the network effect). Thus, the assumption on $\lambda$, that it is no more than the expected revenue from any one user clicking on an ad, keeps the penalty term small, since in practice click-through probabilities tend to be small. Secondly, the regret bound given by the theorem can be understood as follows. Upon termination, the budget-regret from Greedy's allocation is at most $(1/2)p_{max}B$ (plus a small constant $\lambda/2$). 
The theorem says that Greedy achieves such a budget-regret while being frugal w.r.t. the number of seeds it uses. Indeed, its seed-regret is bounded by the minimum number of needs that an optimal algorithm would use to reach the budget, multiplied by a logarithmic factor. 

\begin{proof}[of Theorem~\ref{thm:full-regret}]

We establish a series of claims.
\begin{claim} \label{claim1}
Suppose $S_i$ is the seed set allocated to advertiser $a_i$ and $\Pi_i(S_i) < B_i$.
Then the greedy algorithm will add a node $x$ to $S_i$ iff $|\Pi_i(S_i\cup\{x\}) - B_i| < |\Pi_i(S_i) - B_i|$ and $x = \argmax_{w\in V\setminus S_i} (|\Pi_i(S_i) - B_i| - |\Pi_i(S_i\cup\{w\}) - B_i|)$, with ties broken arbitrarily.
\end{claim}

{\sc Proof of Claim}: Let $x$ be a node such that its addition to $S_i$ strictly reduces the budget-regret and it results in the greatest reduction in budget-regret, among all nodes outside $S_i$. The contribution of every node outside $S_i$ to the seed regret (i.e., the penalty term) is the same and is equal to $\lambda$. Thus, any node that achieves the maximum budget-regret reduction will have the maximum overall regret reduction. Furthermore, the overall regret reduction of adding such a node $x$ to $S_i$ will be non-negative, since its contribution to budget-regret reduction is at least $1\cdot\delta(u,i)\cdot cpe(u,i) \ge \lambda$. So Greedy will add such a node $x$ to $S_i$.
Attention bound does not constrain this addition in anyway since $\kappa_u \ge h$, $\forall u$. \\
($\Longrightarrow$): Let $x$ be the node added by Greedy to $S_i$. By definition, the addition of $x$ to $S_i$ results in a non-negative reduction in overall regret and it leads to the maximum overall regret reduction. By the argument in the ``If'' direction, $x$ must also lead to the maximum reduction in the budget-regret, since seed-regret cannot discriminate between nodes. We will show that this reduction is strictly positive. Since Greedy added $x$ to $S_i$, we have
$\regret(S_i\cup\{x\}) = |\Pi_i(S_i\cup\{x\}) - B_i| + \lambda \cdot (|S_i|+1) \le |\Pi_i(S_i) - B_i| + \lambda \cdot (|S_i|) = \regret(S_i)$. \\
$\Longrightarrow |\Pi_i(S_i\cup\{x\}) - B_i|  \le  |\Pi_i(S_i) - B_i| - \lambda$,  that is, \\ $
|\Pi_i(S_i\cup\{x\}) - B_i|  <  |\Pi_i(S_i) - B_i|$. This was to be shown.  \qed

\begin{claim}\label{claim2}
The budget-regret of Greedy for advertiser $a_i$, upon termination, is at most $(p_i B_i + \lambda)/2$.
\end{claim}


{\sc Proof of Claim}: Consider any iteration $j$. Let $x$ be the seed allocated to advertiser $a_i$ in
this iteration. The following cases arise.

\noindent
$\bullet$ \underline{Case 1}: $\Pi_i(S_i\cup\{x\}) < p_i B_i$.
By submodularity, for any node $y \in V\setminus (S_i\cup\{x\}): MG_i(y|S_i\cup\{x\}) \le MG_i(y|\emptyset) \le p_i B_i$. 
Thus, from Claim~\ref{claim1}, we know the algorithm will continue adding seeds to $S_i$ until Case 2 (below) is reached.

\noindent
$\bullet$ \underline{Case 2}: $\Pi(S_i\cup\{x\}) \ge p_i B_i$.

$\bullet\;$ \underline{Case 2a}: $\Pi(S_i\cup\{x\}) < B_i$.
If $x$ is the last seed added to $S_i$, then $\forall y \in V\setminus (S_i\cup\{x\}):
B_i - \Pi(S_i\cup\{x\}) + \lambda(|S_i|+1) < \Pi_i(S_i\cup\{x\}\cup\{y\}) - B_i + \lambda(|S_i|+2)$.
Notice that upon adding any such $y$, a cross-over must occur w.r.t. $B_i$: suppose otherwise, then adding $y$ would cause net drop in regret and the algorithm would just add $y$ to $S_i\cup\{x\}$, a contradiction.
Simplifying, we get
$B_i - \Pi_i(S_i\cup\{x\}) < \Pi_i(S_i\cup\{x\}\cup\{y\}) - B_i + \lambda$.
Also by submodularity, we have $\Pi_i(S_i\cup\{x\}\cup\{y\}) - \Pi_i(S_i\cup\{x\}) \le p_i B_i$. Thus,\\
$\Longrightarrow \Pi_i(S_i\cup\{x\}\cup\{y\}) - B_i + B_i - \Pi_i(S_i\cup\{x\}) \le p_i B_i$. \\
$\Longrightarrow  2(B_i - \Pi_i(S_i\cup\{x\})) - \lambda \le p_i B_i$. \\
$\Longrightarrow B_i - \Pi_i(S_i\cup\{x\}) \le (p_i B_i + \lambda)/2$.

$\bullet\;$ \underline{Case 2b}: $\Pi_i(S_i\cup\{x\}) > B_i$.
Since Greedy just added $x$ to $S_i$, we infer that 
%
$\Pi_i(S_i) < B_i$ and
%
$[B_i - \Pi_i(S_i)] + \lambda|S_i|  \ge  \Pi_i(S_i\cup\{x\}) - B_i + \lambda(|S_i|+1)$. \\
$\Longrightarrow B_i - \Pi_i(S_i) \ge \Pi_i(S_i\cup\{x\}) - B_i + \lambda$.
Clearly, $x$ must be the last seed added to $S_i$, as any future additions will
strictly raise the regret. By submodularity, we have

$\Pi_i(S_i\cup\{x\}) - \Pi_i(S_i) \le p_i B_i$. \\
$\Longrightarrow \Pi_i(S_i\cup\{x\}) - B_i + B_i - \Pi_i(S_i) \le p_i B_i$. \\
$\Longrightarrow 2(\Pi_i(S_i\cup\{x\}) - B_i) + \lambda \le p_i B_i$. \\
$\Longrightarrow \Pi_i(S_i\cup\{x\}) - B_i \le (p_i B_i - \lambda)/2$.

By combining both cases, we conclude that the budget-regret of
Greedy for $a_i$ upon termination is $\le (p_i B_i + \lambda)/2$. \qed

\smallskip
Next, define $\eta_0 = B_i$.
Let $S_i^j$ be the seed set assigned to advertiser $a_i$ by Greedy after iteration $j$.
Let $\eta_j := \eta_0 - \Pi_i(S_i^j)$, i.e., the shortfall of the achieved revenue w.r.t.
the budget $B_i$, after iteration $j$, for advertiser $a_i$.

\begin{claim}\label{claim3}
After iteration $j$, $\exists x \in V \setminus S_i^j: \Pi_i(S_i\cup\{x\}) - \Pi_i(S_i) \ge 1/s^i_{opt} \cdot \eta_j$,
where $s^i_{opt}$ is the minimum number of seeds needed to achieve a revenue no less than $B_i$.
\end{claim}


{\sc Proof of Claim}: Suppose otherwise. Let $S_i^*$ be the seeds allocated to advertiser $a_i$ by the optimal
algorithm for achieving a revenue no less than $B_i$. Add seeds in $S_i^* \setminus S_i^j$ one by
one to $S_i^j$. Since none of them has a marginal gain w.r.t. $S_i$ that is $\ge 1/s^i_{opt} \cdot \eta_j$, it follows by
submodularity that $\Pi_i(S^j_i \cup S_i^*) \le \Pi(S_i^j) + s^i_{opt} \cdot  1/s^i_{opt} \cdot \eta_j < B_i$,
a contradiction.  \qed

It follows from the above proof that $\eta_j \le \eta_{j-1} \cdot (1 - 1/s^i_{opt})$, which implies that
$\eta_j \le 1/\eta_{j-1} \cdot e^{-1/s^i_{opt}}$. Unwinding, we get $\eta_j \le \eta_0 \cdot e^{-j/s^i_{opt}}$.
Suppose Greedy stops in $\ell$ iterations.  We showed above that the budget-regret of Greedy, for advertiser $a_i$, at the end of this iteration, is either at most
$(p_i \cdot B_i + \lambda)/2$ or is at most $(p_i B_i - \lambda)/2$ depending on the case
that applies. Of these, the latter is more stringent w.r.t. the \#iterations Greedy will take, and
hence w.r.t. the \#seeds it will allocate to $a_i$. So, in iteration $\ell-1$, we have
$\eta_{\ell-1} \ge (p_i B_i - \lambda)/2$. That is, \\
$\eta_{\ell-1} = B_i \cdot e^{-(\ell-1)/s^i_{opt}} \ge (p_i B_i - \lambda)/2$, or \\
$\Longrightarrow  e^{-(\ell-1)/s^i_{opt}} \ge (p_i - \lambda/B_i])/2$. \\
$\Longrightarrow \ell \le 1 + s^i_{opt} \cdot \lceil \ln \{1/(p_i/2 - \lambda/2B_i)\} \rceil$.
Notice that this is an upper bound on $|S_i^\ell|$.
We just proved

\begin{claim} \label{claim4}
When Greedy terminates, the seed-regret for advertiser $a_i$, upon termination, is at most
$\lambda \cdot (1 + s^i_{opt} \cdot \lceil \ln \{1/(p_i/2 - \lambda/2B_i)\} \rceil)$. \qed
\end{claim}

Combining all the claims above, we can infer that the overall regret of Greedy upon termination is
at most
$\sum_{i=1}^h (p_i B_i + \lambda)/2 + \lambda \sum_{i=1}^h [1 + s^i_{opt} (1 + \lceil \ln \{1/(p_i/2 - \lambda/2B_i)\} \rceil$.
\qed
\end{proof}


\subsection{The Case of $\lambda=0$}

In this subsection, we focus on the regret bound achieved by Greedy in the special case that $\lambda=0$, i.e., the overall regret is just the budget-regret. While the results here can be more or less seen as special cases of Theorem~\ref{thm:full-regret}, it is illuminating to restrict attention to this special case. Our first result follows.

\begin{theorem} \label{thm:greedyOneThird}
Consider an instance of \MRA that admits a seed allocation whose total
regret is bounded by a third of the total budget. Then
Algorithm~\ref{alg:rmVanilla} outputs an allocation $\cals$ with a total
regret $\calr(\cals) \leq \frac{1}{3} \cdot B$, where $B = \sum_{i=1}^h
B_i$ is the total budget. 
\end{theorem}

\begin{proof}
Consider an arbitrary iteration of Algorithm~\ref{alg:rmVanilla},
where the algorithm assigns a node, say $u$, to advertiser $a_i$, i.e.,
it adds $u$ to the seed set $S_i$. In particular, notice that $u$ has
been assigned to $<\kappa_u$ advertisers before this iteration, where $\kappa_u$
is the attention bound of $u$.
Three cases arise as shown in Figure
\ref{fig:proof_helper}.

\smallskip\noindent
\underline{Case 1}: $\frac{2}{3} B_i \leq \Pi_i(S_i
\cup \{u\}) \leq \frac{4}{3} B_i$.  In this case, clearly, the regret
for this advertiser is
$|\Pi_i(S_i \cup \{u\})-B_i| \le \min\{\frac{4}{3}B_i -
B_i, B_i - \frac{2}{3}B_i\} \le \frac{1}{3}B_i$.


\smallskip\noindent\underline{Case 2}: $\Pi_i(S_i \cup \{u\}) <
\frac{2}{3} B_i$.  Consider the next iteration in which another seed,
say $u'$, is assigned to the same advertiser $a_i$, i.e., $u'$ is to
$S_i$. Clearly, the marginal gain of $u'$ w.r.t. $S_i\cup\{u\}$ cannot
be more than $\frac{2}{3}B_i$, by submodularity.  Thus, $\Pi_i(S_i \cup
\{u, u'\}) < \frac{4}{3} B_i$. Now, if $\Pi_i(S_i \cup \{u, u'\}) \ge
\frac{2}{3} B_i$, then by Case 1, we have that the regret of advertiser
$a_i$ is at most $\frac{1}{3} B_i$.  Otherwise, $\Pi_i(S_i \cup \{u,
u'\}) < \frac{2}{3} B_i$, and then it is similar to Case 2 condition,
where $u'$ is also added to $S_i$ after $u$. In this case, subsequent
iterations of the algorithm grow $S_i$ until Case 1 is satisfied.  A
simple inductive argument shows that the regret for advertiser $a_i$ is
no more that $\frac{1}{3} B_i$.

\smallskip\noindent\underline{Case 3}: $\Pi_i(S_i \cup \{u\}) >
\frac{4}{3} B_i$.  The algorithm adds $u$ to $S_i$ only when $\Pi_i(S_i
\cup \{u\}) - B_i < B_i - \Pi_i(S_i)$, which implies $\Pi_i(S_i \cup
\{u\}) + \Pi_i(S_i) < 2 B_i$.\footnote{Since the algorithm makes the
choice with lesser regret, we can assume w.l.o.g. that it adds $u$ only
when the addition will result in strictly lower regret than not adding
it.}  However, since $\Pi_i(S_i \cup \{u\}) > \frac{4}{3} B_i$, this
implies $\Pi_i(S_i) < \frac{2}{3}B_i$.  This means the marginal gain of
$u$ w.r.t. $S_i$, i.e., $\Pi_i(S_i\cup\{u\}) - \Pi_i(S_i)$, is larger
than $\frac{2}{3}B_i$. However, $\Pi_i(S_i) < \frac{2}{3}B_i$, which by
submoduarity, implies no subsequent seed can have a marginal gain of
$\frac{2}{3}B_i$ or more, a contradiction. Thus, Case 3 is impossible.

\smallskip
We just showed that for any advertiser, the regret achieved by the
algorithm is at most $\frac{1}{3}B_i$. Summing over all advertisers, we
see that the overall regret is no more than $\frac{1}{3} B$.
\end{proof}


\begin{figure}[t]
  \centering
    \includegraphics[width=0.25\textwidth]{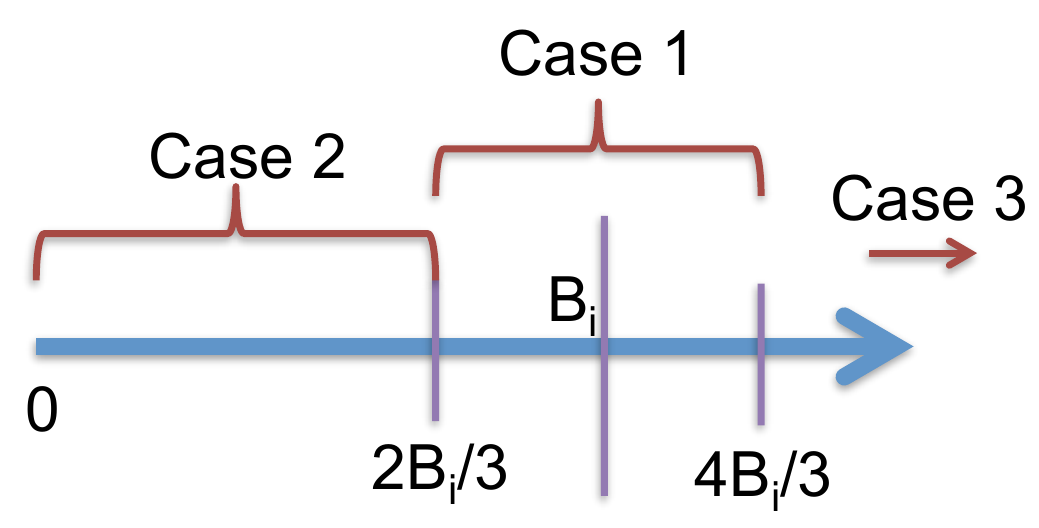}
   \caption{Interpretation of Theorem~\ref{thm:greedyOneThird}.}
   \vspace{-4mm}
\label{fig:proof_helper}

\end{figure}

The regret bound established above is conservative, and unlike Theorem~\ref{thm:full-regret}, does not make any assumptions about the marginal gains of seed nodes. In
practice, as previously noted, most real networks tend to have low influence probabilities and consequently, the marginal gain of any single node tends to be a small fraction of the budget. Using this, we can establish a tighter bound on the regret achieved by Greedy.

\eat{
the cost-per-engagement for any advertiser is usually a small
fraction of the advertiser's budget. Secondly, the marginal gain of any
single node, measured in terms of spread, i.e., $\sigma_i(S_i\cup\{u\})
- \sigma_i(S_i)$, for any seed set $S_i$ and node $u$, is typically a
small fraction of the spread achieved by the final seed set allocated to
advertiser $a_i$. Finally, we know this from the fact that in real networks,
influence probabilities are usually quite low. Motivated by these
observations, we next derive a tighter bound on the regret.

For each advertiser $a_i$, let $p_i$ denote the maximum marginal gain
any seed can provide for the item $i$ divided by budget $B_i$, i.e.,
$p_i = \max_{u\in V} \Pi_i(\{u\})/B_i$. Notice that by submodularity,
marginal gain is the greatest w.r.t. the empty seed set and
$\Pi_i(\emptyset) = 0$.
We assume that the budgets $B_i$ are such that $p_i\in (0,1)$.
Similarly, define $p_{max} = \max_{i=1}^h p_i$.
Next, we give a tighter
approximation bound (w.r.t.\ total budget) for the greedy algorithm.
}

\begin{theorem}\label{thm:greedyBetter}
On any input instance that admits an allocation with total regret
bounded by $\min\{\frac{p_{max}}{2}, 1-p_{max}\} \cdot B$,
Algorithm~\ref{alg:rmVanilla} delivers an allocation \cals\ so
that $\calr(\cals) \leq \min\{\frac{p_{max}}{2}, 1-p_{max}\} \cdot B$.
\end{theorem}

\begin{proof}
The proof is similar to the proof of Theorem~\ref{thm:greedyOneThird}.
Consider an arbitrary iteration of Algorithm~\ref{alg:rmVanilla}.
Suppose $u$ is the seed that the algorithm assigned to $a_i$ (i.e.,
added to seed set $S_i$) in this iteration.  The following two cases
arise.

\smallskip
\noindent
\underline{Case 1}: $\Pi_i(S_i \cup \{u\}) < p_i \cdot B_i$.
Then, the algorithm will continue to add seeds to the seed set
$S_i$, until the condition of Case 2 is met.

\smallskip
\noindent
\underline{Case 2}: $\Pi_i(S_i \cup \{u\}) \ge p_i \cdot B_i$.
There can be two sub-cases in this scenario:

\smallskip
\noindent
\underline{Case 2a}: $\Pi_i(S_i \cup \{u\}) \le B_i$.
Clearly, regret is
\begin{align*}
	B_i - \Pi_i(S_i \cup \{u\}) \le B_i - p_i \cdot B_i = (1-p_i)B_i.
\end{align*}

If $u$ is the last seed added to the seed set $S_i$, then we have regret
$\le (1 - p_i) B_i $.  Moreover, $u$ being the last seed also implies
that for any other node $u' \not\in S_i$, we have
\begin{align*}
B_i - \Pi_i(S_i \cup \{u\}) \le \Pi_i(S_i \cup \{u, u'\}) - B_i,
\end{align*}
since otherwise, the algorithm would have added $u'$ to $S_i$ to
decrease the regret.
Also, due to submodularity,
\begin{align*}
	&\Pi_i(S_i \cup \{u, u'\}) - \Pi_i(S_i \cup \{u\}) \le p_i \cdot B_i , \\
\Longrightarrow \quad &\Pi_i(S_i \cup \{u,u'\}) - B_i + B_i - \Pi_i(S_i \cup \{u\}) \le p_i \cdot B_i , \\
\Longrightarrow \quad	&2 \cdot (B_i - \Pi_i(S_i \cup \{u\})) \le p_i \cdot B_i,  \\
\Longrightarrow \quad	&B_i -\Pi_i(S_i \cup \{u\}) \le \frac{p_i}{2} \cdot B_i.
\end{align*}

Therefore, in Case 2a, if  $u$ is the last seed selected by the
algorithm, then regret of advertiser $a_i$ is $\min\{\frac{p_i}{2},
1-p_i\} \cdot B_i$.  Otherwise, the algorithm would continue with the
next iteration and add seeds until Case 2a or Case 2b is satisfied.

\smallskip
\noindent
\underline{Case 2b}: $\Pi_i(S_i \cup \{u\}) > B_i$.
Then regret for advertiser $a_i$ is $\Pi_i(S_i \cup \{u\}) - B_i$.

In this case, $u$ must be the last seed selected by the algorithm as
adding another seed can only increase the regret.
Therefore, it is clear that
\begin{align*}
\Pi_i(S_i \cup \{u\}) - B_i \le B_i - \Pi_i(S_i).
\end{align*}

Moreover, due to submodularity, we know that
\begin{align*}
	& \Pi_i(S_i \cup \{u\}) - \Pi_i(S_i) \le p_i \cdot B_i, \\
\Longrightarrow \quad	& \Pi_i(S_i \cup \{u\}) - B_i + B_i - \Pi_i(S_i) \le p_i \cdot B_i, \\
\Longrightarrow \quad	&2 \cdot (\Pi_i(S_i \cup \{u\}) - B_i) \le p_i \cdot B_i ,\\
\Longrightarrow \quad	&\Pi_i(S_i \cup \{u\}) - B_i \le \frac{p_i}{2} \cdot B_i .
\end{align*}

Combining Cases 2a and 2b, and summing it over all advertisers, it is
easy to see that total regret is $\le \min(\frac{p_{max}}{2},
(1-p_{max})) \cdot B$.
\end{proof}

We note that this claim generalizes Theorem~\ref{thm:greedyOneThird}. In
fact, the two bounds: $\frac{p_{max}}{2}$ and $1-p_{max}$ meet at the
value of $1/3$ when $p_{max} = 2/3$. In practice, $p_{max}$ may be much
smaller, making the bound better. 

\section{Scalable Algorithms}
\label{sec:algo}
Algorithm~\ref{alg:rmVanilla} (Greedy) involves a large number of calls to influence spread computations, to find the node for each advertiser $a_i$ that yields the maximum decrease in regret $\mathcal{R}_i(S_i)$.
Given any seed set $S$, computing its {\em exact} influence spread $\sigma(S)$ under the IC model is \SPhard~\cite{ChenWW10}, and this hardness trivially carries over to the topic-aware IC model~\cite{BarbieriBM12} with CTPs.
A common practice is to use Monte Carlo (MC) simulations to estimate influence spread~\cite{kempe03}.
However, accurate estimation requires a large number of MC simulations, which is prohibitively expensive and not scalable. Thus, to make Algorithm~\ref{alg:rmVanilla} scalable, we need an alternative approach. 

In the influence maximization literature, considerable effort has been devoted to developing more efficient and scalable algorithms~\cite{ChenWW10,jung12, borgs14, tang14, cohen14}.
\WL{Of these, the IRIE algorithm proposed by Jung et al.~\cite{jung12} is a state-of-the-art heuristic for influence maximization under the IC model and is orders of magnitude faster than MC simulations. We thus use a variant of Greedy, \irie, where IRIE replaces MC simulations for spread estimation. It is one of the strong baselines we will compare our main algorithm with in \textsection\ref{sec:exp}.
In this section, we instead propose a scalable algorithm with guaranteed approximation for influence spread.}  

Recently, Borgs et al.~\cite{borgs14} proposed a quasi-linear time randomized algorithm based on the idea of sampling \emph{``reverse-reachable''} (RR) sets in the graph.
It was improved to a near-linear time randomized algorithm -- {\em Two-phase Influence Maximization (TIM)} -- by Tang et al.~\cite{tang14}.
Cohen et al.~\cite{cohen14} proposed a sketch-based design for fast computation of influence spread, achieving efficiency and effectiveness comparable to TIM.
We choose to extend TIM 
as it is the current state-of-the-art influence maximization algorithm and is more adapted to our needs. 

In this section, we adapt the essential ideas from \WL{Greedy}, RR-sets sampling, and the TIM algorithm to devise an algorithm for \MRA, called \fastAlgorithm (\fastAlgo for short), that is much more efficient and scalable than Algorithm~\ref{alg:rmVanilla} with MC simulations. \LL{Our adaptation to TIM is non-trivial, since TIM relies on knowing the exact number of seeds required. In our framework, the number of seeds needed is driven by the budget and the current regret and so is dynamic.} 
We first give the background on RR-sets sampling, review the TIM algorithm~\cite{tang14}, and then describe our \fastAlgo algorithm.

\subsection{Reverse-Reachable Sets and TIM}

\noindent\textbf{RR-sets Sampling: Brief Review.}
We first review the definition of RR-sets, which is the backbone of both TIM and our proposed \fastAlgo algorithm.
Conceptually speaking, a random RR-set $R$ from $G$ is generated as follows.
First, for every edge $(u,v) \in E$, remove it from $G$ w.p.\ $1-p_{u,v}$: this generates a possible world $X$.
Second, pick a \emph{target} node $w$ uniformly at random from $V$.
Then, $R$ consists of the nodes that can reach $w$ in $X$.
\WL{This can be implemented efficiently by first choosing a target node $w\in V$ uniformly at random and performing a breadth-first search (BFS) starting from it.
Initially, create an empty BFS-queue $Q$, and insert all of $w$'s in-neighbors into $Q$.
The following loop is executed until $Q$ is empty:
Dequeue a node $u$ from $Q$ and examine its {\em incoming} edges: for each edge $(v,u)$ where $v\in N^{in}(u)$, we insert $v$ into $Q$ w.p. $p_{v,u}$. 
All nodes dequeued from $Q$ thus form a RR-set. }

The intuition behind RR-sets sampling is that, if we have sampled sufficiently many RR-sets, and a node $u$ appears in a large number of RR sets, then $u$ is likely to have high influence spread in the original graph and is a good candidate seed.

\smallskip\noindent\textbf{TIM: Brief Review.}
Given an input graph $G=(V,E)$ with influence probabilities and desired seed set size $s$,
TIM, in its first phase, computes a lower bound on the optimal influence spread of any seed set of size $s$, i.e., $OPT_s := \max_{S\subseteq V, |S| = s} \sigic(S)$. Here $\sigic(S)$ refers to the spread w.r.t. classic IC model.  TIM then uses this lower bound to estimate the number of random RR-sets that need to be generated, denoted $\theta$.
In its second phase, TIM simply samples $\theta$ RR-sets, denoted $\RR$, and uses them to select $s$ seeds, by solving the Max $s$-Cover problem: find $s$ nodes, that between them, appear in the maximum number of sets in $\RR$. This is solved using a well-known greedy procedure: start with an empty set and repeatedly add a node that appears in the maximum number of sets in $\RR$ that are not yet ``covered''.
\eat{ running a simple greedy procedure:
We start with $S = \emptyset$.
In each iteration, we find the node $u$ that appears in (covers) the most of the sets in $\RR$;
Then add $u$ to $S$ and remove the sets covered by $u$ from $\RR$.
The selection is done when $|S| = s$. }

TIM provides a $(1-1/e-\epsilon)$-approximation to the optimal solution $OPT_s$ with high  probability.
Also, its time complexity is $O((s+\ell)(|V|+|E|)\log|V|/\epsilon^2)$, while that of the greedy algorithm (for influence maximization) is $\Omega(k|V||E|\cdot \mathrm{poly}(\epsilon^{-1}))$.


\smallskip\noindent\textbf{Theoretical Guarantees of TIM.}
Consider any collection of random RR-sets, denoted $\RR$.
Given any seed set $S$, we define $F_{\RR}(S)$ as the fraction of $\RR$ covered by $S$, where $S$ covers an RR-set iff it overlaps it.
The following proposition says that for any $S$, $|V| \cdot F_{\RR}(S)$ is an unbiased estimator of $\sigic(S)$.  

\begin{proposition}[Corollary 1, \cite{tang14}]\label{prop:coro1}
Let $S\subseteq V$ be any set of nodes, and $\RR$ be a collection of random RR sets.
Then, $\sigic(S) = \E[ |V| \cdot F_\RR(S)]$.
\end{proposition}

The next proposition shows the accuracy of influence spread estimation and the approximation gurantee of TIM.
Given any seed set size $s$ and $\varepsilon > 0$, define $L(s,\varepsilon)$ to be:
\begin{align}\label{eqn:timLB}
L(s, \varepsilon) = (8 + 2 \varepsilon) n \cdot \dfrac{\ell \log n + \log \binom{n}{s} + \log 2}{OPT_s \cdot \varepsilon^{2}},
\end{align}
where $\ell > 0, \epsilon > 0$.

\begin{proposition}[Lemma 3 \& Theorem 1, \cite{tang14}]
\label{lemma:TIMLemma3}
Let $\theta$ be a number no less than $L(s, \varepsilon)$.
Then for any seed set $S$ with $|S| \leq s$, the following inequality holds w.p.\ at least $1 - n^{- \ell} / \binom{n}{s}$:
\begin{align}
\label{eq:Lemma3}
\left| |V| \cdot F_{\RR}(S) - \sigic(S) \right| < \dfrac{\varepsilon}{2} \cdot OPT_s.
\end{align}
Moreover, with this $\theta$, TIM returns a $(1-1/e-\epsilon)$-approximation  to $OPT_s$ w.p. $1-n^{-\ell}$.
\end{proposition}

This result intuitively says that as long as we sample enough RR-sets, i.e., $|\RR| \geq \theta$, the absolute error of using $|V| \cdot F_{\RR}(S)$ to estimate $\sigic(S)$ is bounded by a fraction of $OPT_s$ with high probability.
Furthermore, this gives approximation guarantees for influence maximization.
Next, we describe how to extend the ideas of RR-sets sampling and TIM for regret minimization.

\subsection{Two-phase Iterative Regret Minimization}
A straightforward application of TIM for solving \MRA will not work.
There are two critical challenges.
First, TIM requires the number of seeds $s$ as input, while the input of \MRA is in the form of monetary budgets, and thus we do not know the precise number of seeds that should be allocated to each advertiser beforehand.
Second, our influence propagation model has click-through probabilities (CTPs) of seeds, namely $\delta(u,i)$'s.
This is not accounted for in the RR-sets sampling method: it implicitly assumes that each seed becomes active w.p.\ $1$.

We first discuss how to adapt RR-sets sampling to incorporate CTPs.
Then we 
deal with unknown seed set sizes.

\eat{
To address the vital gaps for \MRA, we devise an algorithm called {\em Iterative Two-phase Regret Minimization} (ITRM).
At a high level, ITRM uses budgets and the influence spread of some ``representative'' nodes to come up with a series of guesses for the number of seeds needed for each advertiser.
The guesses are made in an iterative, incremental manner.
We first generates a corresponding number of random RR sets based on initial guesses and select seeds accordingly.
Based on the latest selections and current regret, we revise the guess to include more seeds, until all companies are saturated.
}


\spara{RR-sets Sampling with Click-Through Probabilities.}
Recall that in our model, when a node $u$ is chosen as a seed for advertiser $a_i$, it has a probability $\delta(u,i)$ to accept being seeded, i.e., to actually click on the ad.
For ease of exposition, in the rest of this subsubsection only, we assume that there is only one advertiser, and the CTP of each user $u$ for this advertiser is simply $\delta(u) \in [0,1]$.
The technique we discuss and our results readily extend to any number of advertisers.


\blue{
For clarity, we call the RR-sets generated with CTPs incorporated 
	{\em RRC-sets} to distinguish them from normal RR-sets, which have no 
associated CTPs.
The procedure for generating a random RRC-set is similar to that for 
	generating a normal (random) RR-set.
First, a root $w$ is chosen uniformly at random from $V$.
Let $R_w$ denote the associated RRC-set being generated.
Then, we enqueue $w$ into a FIFO queue $Q$.
}

\blue{
Until $Q$ is empty, we repeat the following: dequeue the next node from $Q$, and let it be $u$.
For all of its in-neighbors $v\in N^{in}(u)$, we first test the edge $(v,u)$:
	it is live w.p.\ $p_{v,u}$, and blocked w.p.\ $1-p_{v,u}$.
If the edge is blocked, we ignore it and continue to the 
next in-neighbor, if any. If the edge is live, 
we further flip a biased coin, independently, for the node $v$ itself:
	w.p.\ $\delta(v)$, we declare $v$ live, and w.p.\ $1-\delta(v)$, declare
	$v$ blocked.
The following two cases arise:
$(i)$.\  If $v$ is live, then it can be a valid seed, and thus we add $v$
	to $R_w$ as well as enqueue $v$ into $Q$.
$(ii)$.\ If $v$ is blocked, then it cannot be a valid seed itself, but it should
	still be added to $Q$, since its in-neighbors may still be
	valid seeds, depending on their own edge- and node-based coin flips.	
}

\blue{
Note that for the root $w$ itself, the node test should also be performed using its CTP:
	w.p. $\delta(w)$, $w$ is added to $R_w$.
Again, even if this CTP test fails, which occurs w.p.\ $1-\delta(w)$, the above
	procedure is still correct in terms of first enqueuing $w$ into $Q$, since
	$w$'s in-neighbors can be valid seeds to activate $w$.
}

Let $\RQ$ be a collection of RRC-sets.
Similar to $F_\RR(S)$, for any set $S$, we define $F_\RQ(S)$ to be the fraction of $\RQ$ that overlap with $S$.
\LL{Let $\sigicctp(S)$ be the influence spread of a seed set $S$ under the IC model with CTPs.} 
We first establish a similar result to Proposition~\ref{prop:coro1} which says that $|V| F_\RQ(S)$ is an unbiased estimator of $\sigma(S)$.

\begin{lemma}\label{lemma:equiv}
Given a graph $G=(V,E)$ with influence probabilities on edges, for any $S \subseteq V$,
$\sigicctp(S) = \E[ |V| \cdot F_\RQ(S)].$

\end{lemma}

\begin{proof}
We show the following equality holds:
\begin{align}\label{eqn:ctr-inf}
\sigicctp(S)/|V| = \E[F_\RQ(S)].
\end{align}
The LHS of \eqref{eqn:ctr-inf} equals the probability that a node chosen uniformly at random can be activated by seed set $S$ where a seed $u\in S$ may become live with CTP $\delta(u)$, while the RHS of \eqref{eqn:ctr-inf} equals the probability that $S$ intersects with a random RRC-set.
They both equal the probability that a randomly chosen node is reachable by $S$ in a possible world corresponding to the IC-CTP model.
\end{proof}

\eat{\note[Wei]{I assume that we have described the possible world under IC-CTP model and we have also stated the theorem and proof Laks gave for IC-CTP model.  Laks: Since you were editing the earlier sections lately, I assume you took care of this.}
}

In principle, RRC-sets are those we should work with for the purpose of seed selection for \MRA.
However, note that by Equation~\eqref{eqn:timLB} and Proposition~\ref{lemma:TIMLemma3}, the number of samples required is inversely proportional to the value of the optimal solution $OPT_s$.
However, in reality, click-through rates on ads are quite low, and thus $OPT_s$, taking CTPs into account, will decrease by at least two orders of magnitude (e.g., $OPT_s$ with CTP $0.01$ would become 100 times smaller than $OPT_s$ with CTP $1$).
This in turn translates into at least two orders of magnitude more RRC-sets to be sampled, which ruins scalability.

An alternative way of incorporating CTPs is to pretend as though all CTPs were $1$.
We still generate RR-sets, and use the estimations given by RR-sets to compute revenue.
More specifically, for any $S\subseteq V$ and any $u\in V\setminus S$, we compute the marginal gain of $u$ w.r.t.\ $S$, namely $\sigma_C(S \cup \{u\}) - \sigma_C(S)$, by $\delta(u) \cdot |V| \cdot [F_\RR(S \cup \{u\}) - F_\RR(S)]$.
This avoids sampling of numerous RRC-sets.

We can show that in expectation, computing marginal gain in IC-CTP model using RRC-sets is essentially equivalent to computing it under the IC model using RR-sets in the manner above.

\begin{theorem}\label{thm:ctps}
Consider any $u\in S$ and any $S\subseteq V$.
Let $\delta(u)$ be the probability that $u$ accepts to become a seed.
Let $\RR$ and $\RQ$ be a collection of RR-sets and of RRC-sets, respectively.
Then,
\begin{align*}
\delta(u)(\E[F_\RR(S \cup \{u\})] - \E[F_\RR(S)]) = \E[F_\RQ(S \cup \{u\})] - \E[F_\RQ(S)].
\end{align*}
\end{theorem}

\begin{proof}
Consider a random RR-set $X$, and define an indicator function $\mathbb{I}_X(u, S)$, which takes on $1$ \emph{if} $u \in X$ {\em and} $S \cap X = \emptyset$, and $0$ \emph{otherwise}.
Then, we have:
\begin{align} \label{eqn:rr1}
&\E[F_\RR(S \cup \{u\})] - \E[F_\RR(S)] \nonumber \\
&= \sum_X \Pr[X] \cdot \mathbb{I}_X(u, S) = \sum_{X \colon \mathbb{I}_X(u, S) = 1} \Pr[X],
\end{align}

where $\Pr[X]$ is the probability of sampling the RR-set $X$.

\blue{
Note that the only difference between the generation of
	an RR-set and that of an RRC-set is the additional coin flips 
	on nodes, with CTPs, which are all independent.
Now, consider a fixed RR-set $X$ that does contain $u$.
If we were to generate an RRC-set --- meaning that the outcomes
	of all edge-level coin flips would remain the same --- then $X$
	may contain $u$ w.p.\ $\delta(u)$.
This is true since all edge- and node-level coin flips are
	independent.
If  $u$ belongs to the RRC-set realization of $X$, we denote it by $X_u$.
}


Now, for the expected marginal gain of $u$ under the model
	with CTPs, we have:
\begin{align*}
&\E[F_\RQ(S \cup \{u\})] - \E[F_\RQ(S)] \\
&= \sum_{X_u} \Pr[X_u] = \sum_{X: \mathbb{I}_X(u, S) = 1} \delta(u) \cdot \Pr[X] \\
&= \delta(u) \cdot (\E[F_\RR(S \cup \{u\}) - \E[F_\RR(S)]),
\end{align*}
where we have applied \eqref{eqn:rr1} in the last equality.
This completes the proof.
\end{proof}

This theorem shows even with CTPs, we can still use the usual RR-sets sampling process for estimating spread efficiently and accurately as long as we multiply marginal gains by CTPs. {\sl This result carries over to the setting of multiple advertisers.} 

\spara{Iterative Seed Set Size Estimation.}
As mentioned earlier, TIM needs the required number of seeds $s$ as input, which is not available for the \MRA problem. From the advertiser budgets, there is no obvious way to determine the number of seeds. This poses a challenge since the required number of RR-sets ($\theta$) depends on $s$.
To circumvent this difficulty, we propose a framework which first makes an initial guess at $s$, and then iteratively revises the estimated value, until no more seeds are needed, while concurrently selecting seeds and allocating them to advertisers.

For ease of exposition, let us first consider a single advertiser $a_i$. Let $B_i$ be the budget of $a_i$ and let $s_i$ be the true number of seeds required to minimize the regret for $a_i$. We do not know $s_i$ and estimate it in successive iterations as $\tilde{s}_i^t$.
We start with an estimated value for $s_i$, denoted $\tilde{s_i}^1$, and use it to obtain a corresponding ${\theta}_i^1$ ({\em cf.} Proposition~\ref{lemma:TIMLemma3}).
If ${\theta}_i^t > {\theta}_i^{t-1}$,\footnote{Assuming $\theta_i^0 = 0, i = 1, \ldots, h$.} we will need to sample an additional $({\theta}_i^t - {\theta}_i^{t-1})$ RR-sets, and use all RR-sets sampled up to this iteration to select $(\tilde{s}_i^t - \tilde{s}_i^{t-1})$ additional seeds.
After adding those seeds, if $a_i$'s budget $B_i$ is not yet reached, this means more seeds can be assigned to $a_i$.
Thus, we will need another iteration and we further revise our estimation of $s_i$.
The new value, $\tilde{s}_i^{t+1}$, is obtained by adding to $\tilde{s}_i^t$
the floor function of the ratio between the current regret $\regret_i(S_i)$ and the \LL{marginal} revenue contributed by the $\tilde{s}_i^t$-th seed (i.e., the latest seed).
This ensures we do not overestimate, thanks to submodularity, as future seeds have diminishing marginal gains.

\eat{
Let $z$ be the total number of iterations.
Thus, the final estimate for $s_i$ is simply $\hat{s}_i^z$.
Let $\RR_i^j$ be the collection of RR-sets generated in the $j$-th iteration.
To finalize the seed set $S_i$ for $a_i$, we use $\RR_i := \cup_{j=1}^z \RR_i^j$ to select  $\hat{s}_i^z$ seeds.
Note that the spread of this final $S_i$ has bounded error as the size of $\RR_i$ is nothing but the $\theta_i$ corresponding to $\hat{s}_i^z$.
}

\IncMargin{1em}
\begin{algorithm}[t!]
\caption{\fastAlgo}
\label{alg:rmRRsets}
\Indm
{\small
\SetKwInOut{Input}{Input}
\SetKwInOut{Output}{Output}
\SetKwComment{tcp}{//}{}
\Input{$G=(V,E)$; attention bounds $\kappa_u, \forall u\in V$; items $\vec{\gamma}_i$ with $cpe(i)$
\& budget $B_i$, $i = 1,\ldots,h$; CTPs $\delta(u,i), \forall u \forall i$}
\Output{$S_1, \cdots, S_h$}
}
\Indp
{\small
\ForEach{$j = 1, 2, \ldots, h$} {
$S_j \gets \emptyset$; $Q_j \gets \emptyset$; \tcp{\small a priority queue}
$s_j \gets 1$; $\theta_j \gets L(s_j, \varepsilon)$; $\RR_j \gets \mathsf{Sample}(G, \gamma_j,\theta_j)$\;  \label{line:ITRM-1}
}
\BlankLine
\While{true} {
  \ForEach{$j = 1, 2, \ldots, h$} {
  	  $(v_j, cov_j(v_j)) \gets \mathsf{SelectBestNode}(\RR_j)$ \tcp*{Algo~\ref{alg:rrBestNode}}	  \label{line:ITRM-2}
 	  $F_{\RR_j}(v_j) \gets cov_j(v_j) / \theta_j $\;
  }
  $i \leftarrow \argmax_{j=1}^h \mathcal{R}_j(S_j) -\mathcal{R}_j(S_j \cup \{v_j\}) $ \hspace*{12ex} subject to: $\mathcal{R}_j(S_j \cup \{v_j\}) < \mathcal{R}_j(S_j)$;
//{\tt find the (user, ad) pair with max drop in regret.}

\If{$i \neq \mathbf{NULL}$} {
	$S_i \gets S_i \cup \{v_i\}$\;
         $Q_i.\mathsf{insert}(v_i, cov_i(v_i)) $\;
	$\RR_i \gets \RR_i \setminus \{R \mid v_i \in R \; \wedge \; R \in \RR_i\} $;      \label{line:ITRM-3}
}
//{\tt remove RR-sets that are covered}\;
\lElse {
	{\bf return}
}
\If{$ \left\vert{S_i}\right\vert = s_i$} {  \label{line:ITRM-4}
	$s_i \gets s_i + \lfloor \mathcal{R}_i(S_i) / (cpe(i) \cdot n \cdot \delta(v_i,i) \cdot F_{\RR_i}(v_i)) \rfloor$\;	
	$\theta_i \gets \max\{L(s_i,\varepsilon), \theta_{i}\}$\;
  	$\RR_i \gets \RR_i \cup \mathsf{Sample}(G, \gamma_i, \max\{0, L(s_i,\varepsilon) - \theta_i)\}$\; 	
  	$\Pi_i(S_i) \gets$ {\sf UpdateEstimates}($\RR_i$, $\theta_i$, $S_i$, $Q_i$);  \label{line:ITRM-6}
//{\tt revise estimates to reflect newly added RR-sets}\;
  	$\mathcal{R}_i(S_i) \gets |B_i - \Pi_i(S_i)|$\;   \label{line:ITRM-5}
}
}
}
\end{algorithm}
\DecMargin{1em}

\IncMargin{1em}
\begin{algorithm}[t!]
\caption{SelectBestNode($\RR_j$)}
\label{alg:rrBestNode}
\Indm
{\small
\SetKwInOut{Input}{Input}
\SetKwInOut{Output}{Output}
\SetKwComment{tcp}{//}{}
\Output{$(u, cov_j(u))$}
}
\Indp
{ \small
    $u \gets \argmax_{v\in V} |{\{R \mid v \in R \; \wedge \; R \in \RR_j\}}| $
\hspace*{12ex} subject to: $|\{S_l | v \in S_l\}| < \kappa_v $\; \label{algo-sbn:line1}
    $cov_j(u) \gets |{\{R \mid u \in R \; \wedge \; R \in \RR_j\}}| $;
//{\tt find best seed for ad $a_j$ as well as its coverage.}
}

\end{algorithm}
\DecMargin{1em}

\IncMargin{1em}
\begin{algorithm}[t!]
\caption{UpdateEstimates($\RR_i$, $\theta_i$, $S_i$, $Q_i$)}
\label{alg:rrUpdateEst}
\Indm
{\small
\SetKwInOut{Input}{Input}
\SetKwInOut{Output}{Output}
\SetKwComment{tcp}{//}{}
\Output{$\Pi_i(S_i)$}
}
\Indp
{ \small
$\Pi_i(S_i) \gets 0 $ \;
\For{$j = 0, \ldots, |S_i| -1$} {
	$(v,cov(v)) \gets Q_i[j] $ \;
	$cov'(v) \gets \left\vert{\{R \mid v \in R, R \in \RR_i\}}\right\vert $\;
	$Q_i.\mathsf{insert}(v, cov(v) + cov'(v))$\;
	$\Pi_i(S_i) \gets \Pi_i(S_i) + cpe(i) \cdot n \cdot \delta(v,i) \cdot ((cov(v) + cov'(v)) / \theta_i) $; //{\tt update coverage of existing seeds w.r.t. new RR-sets added to collection.}
}
}
\end{algorithm}
\DecMargin{1em}

Algorithm~\ref{alg:rmRRsets} outlines \fastAlgo, which integrates the iterated seed set size estimation technique above, suitably adapted to multi-advertiser setting, along with the RR-set based coverage estimation idea of TIM, and uses Theorem~\ref{thm:ctps} to deal with CTPs. Notice that the core logic of the algorithm is still based on greedy seed selection as outlined in Algorithm~\ref{alg:rmVanilla}.
Algorithm \fastAlgo works as follows.
For every advertiser $a_i$, we initially set its seed budget $s_i$ to be 1 (a conservative, but safe estimate), and find the first seed using random RR-sets generated accordingly (line~\ref{line:ITRM-1}).
In the main loop, we follow the greedy selection logic of Algorithm~\ref{alg:rmVanilla}.
That is, every time, we identify the valid user-advertiser pair $(u,a_i)$ that gives the {\sl largest decrease in total regret} and allocate $u$ to $S_i$ (lines~\ref{line:ITRM-2} to \ref{line:ITRM-3}), paying attention to the attention bound of $u$ (line ~\ref{algo-sbn:line1} of Algorithm~\ref{alg:rrBestNode}).
If $|S_i|$ reaches the current estimate of $s_i$ after we add $u$, then we increase $s_i$ by $\lfloor \mathcal{R}_i(S_i) / (cpe(i) \cdot n \cdot F_{\RR_i}(u)) \rfloor$ (lines~\ref{line:ITRM-4} to \ref{line:ITRM-5}), as described above, as long as the regret continues to decrease. 
Note that after adding additional RR-sets, we should update the spread estimation of current seeds w.r.t.\ the new collection of RR-sets (line~\ref{line:ITRM-6}).
This ensures that future marginal gain computations and selections are accurate.
This is effectively a {\em lower bound} on the number of additional seeds needed, as subsequent seeds will not have marginal gain higher than that of $u$ due to submodularity.
As in Algorithm~\ref{alg:rmVanilla}, \fastAlgo terminates when all advertisers have saturated, i.e., no additional seed can bring down the regret.
\LL{Note that in Algorithm~\ref{alg:rrUpdateEst}, we update the estimated revenue (coverage) of existing seeds w.r.t. the additional RR-sets sampled, to keep them accurate.}

\eat{
\note[Wei]{ALGO 2: Have some doubts about the two red lines.  If we remove RR-sets from $\RR$ due to adding seeds (1st red line), then later (2nd red line), we should add it back right?}

\note[Wei]{ALGO 4: The red line, should it be just $Q_i.insert(v, cov')$?}
}

\eat{First we explain how Iter-RRS can be performed. Let $T$ denote the number of iterations of random $RR$ sets sampling performed in Iter-RRS, let $\kappa^t$ denote the incremental seed set sizes for each sampling iteration $t = 1,\cdots,T$. Let $\theta^t \ge L(\sum_{j=1}^{t}\kappa^j,\varepsilon)$ be the number of random $RR$ sets required for the influence spread estimation of sets of upto $\sum_{j=1}^{Estimation Accuracy oft}\kappa^j$ nodes. Denote $\RR^t$ as the set of random $RR$ sets of size $\theta^t$ obtained at each sampling iteration $t$ from $\RR^t = \RR^{t-1} \cup \RR$, where $\RR$ is the set of additional $\theta^t - \theta^{t-1}$ $RR$ sets sampled at iteration $t$.
}

\eat{
\smallskip\noindent\textbf{Theoretical Analysis for \fastAlgo.}
Note that, when $\theta^t \ge \theta^{t-1}$, no additional sampling can be done, thus, the previous lower bound $L(\sum_{j=1}^{t-1} \kappa^j,\varepsilon)$ is used for $\theta^t$, by defining $\theta^t = \theta^{t-1}$, since $\theta^{t-1} \ge \theta^{t} \ge L(\sum_{j=1}^{t} \kappa^j, \varepsilon)$, \emph{i.e.}, $\theta^{t-1}$ already satisfies the lower bound bound $\theta^t$. We establish the accuracy of Iter-RRS based estimations at any iterative sampling step next.
}

\smallskip\noindent\textbf{Estimation Accuracy of \fastAlgo.}
\LL{At its core, \fastAlgo, like TIM, estimates the spread of chosen seed sets, even though its objective is to minimize regret w.r.t. a monetary budget.} 
Next, we show that the influence spread of seeds estimated by \fastAlgo enjoys bounded error guarantees similar to those chosen by TIM (see Proposition~\ref{lemma:TIMLemma3}).

\begin{theorem}\label{thm:iterRR}
 At any iteration $t$ of iterative seed set size estimation in Algorithm \fastAlgo, for any set $S_i$ of at most $s = \sum_{j=1}^{t}s^j$ nodes, $\left| n \cdot F_{\RR^t}(S_i) - \sigma_i(S_i) \right| < \dfrac{\varepsilon}{2} \cdot OPT_s$ holds with probability at least $1 - n^{-\ell} / \binom{n}{s}$, where $\sigma_i(S)$ is the expected spread of seed set $S_i$ for ad $i$. 
\end{theorem}

\begin{proof}
When $t = 1$, our claim follows directly from Proposition~\ref{lemma:TIMLemma3}.
When $t > 1$, by definition of our iterative sampling process, the number of RR-sets, $|\RR^t|$, is equal to $\max_{j=1,\ldots,t} L_j,$ 
where $L_j = L\left( \sum_{a=1}^{j}s^a, \varepsilon \right)$.
This means that at any iteration $t$, the number of RR-sets is always sufficient for Eq. \eqref{eq:Lemma3} to hold.
Hence, for the set $S_i$ containing seeds accumulated up to iteration $t$, our claim on the absolute error in the estimated spread of $S_i$ holds, by virtue of 
Proposition~\ref{lemma:TIMLemma3}.
\end{proof}

\section{Experiments}
\label{sec:exp}
We conduct an empirical evaluation of the proposed algorithms. The goal is manifold. First, we would like to evaluate the quality of the algorithms as measured by the regret achieved, the number of seeds they used to achieve a certain level of budget-regret, and the extent to which the attention bound ($\kappa$) and the penalty factor ($\lambda$) affect their performance. Second, we evaluate the efficiency and scalability of the algorithms w.r.t. advertiser budgets, which indirectly control the number of seeds required, and w.r.t. the number of advertisers. We measure both running time and memory usage. 


\smallskip\noindent\textbf{Datasets.}
\WL{Our experiments are based on four real-world social networks, whose basic statistics are summarized in Table~\ref{table:dataset}.}
Of the four datasets, we use \flix and \epi for our quality experiments and \dblp and \livej for scalability experiments. 
\flix is from a
social movie-rating site (\url{http://www.flixster.com/}). 
The dataset records movie ratings from users along with their timestamps. 
We use the topic-aware influence probabilities \CA{and the item-specific topic distributions
provided by the authors of \cite{BarbieriBM12}, who learned the probabilities using maximum likelihood estimation for the TIC model with $K=10$ latent topics.
In our quality experiments, we set the number of advertisers $h$ to be $10$, and used $10$ of the learnt topic distributions from Flixster dataset, where for each ad $i$ , its topic distribution $\vec{\gamma_i}$ has mass $0.91$ in the $i$-th topic, and $0.01$ in all others.}
CTPs are sampled uniformly at random from the interval $[0.01, 0.03]$ for all user-ad pairs, in keeping with real-life CTPs (see \textsection\ref{sec:intro}). 

\epi is a who-trusts-whom network taken from a consumer review website (\url{http://www.epinions.com/}).
\CA{For Epinions, we similarly set $h=10$ and use $K=10$ latent topics. For each ad $i$, we use synthetic topic distributions $\vec{\gamma_i}$, by borrowing the ones used in \flix. 
For all edges and topics, the topic-aware influence probabilities are sampled from an exponential distribution with mean $30$, via the inverse transform technique~\cite{devroye1986sample} on the values sampled randomly from uniform distribution $\mathcal{U}(0,1)$. 
}

\begin{table}[t!]
\small
\centering
\begin{tabular}{|c | c | c | c | c| }
\hline
& \flix & \epi  & \dblp & \livej  \\ \hline
\#nodes & 30K & 76K & 317K  & 4.8M \\ \hline
\#edges & 425K & 509K & 1.05M & 69M  \\ \hline
type & directed & directed & undirected & directed \\ \hline
\end{tabular}
\caption{Statistics of network datasets.}
\label{table:dataset}
\vspace{-1mm}
\end{table}

For scalability experiments, we adopt two large networks \dblp and \livej (both are available at \url{http://snap.stanford.edu/}).
\dblp is a co-authorship graph (undirected) where nodes represent authors and there is an edge between two nodes if they have co-authored a paper indexed by DBLP.
\WL{We direct all edges in both directions.}
\livej is an online blogging site where users can declare which other users are their friends.

\WL{In all datasets, advertiser budgets and CPEs are chosen in such a way that the total number of seeds required for all ads to meet their budgets is less than $n$.
This ensures no ads are assigned empty seed sets.
For lack of space, we do not enumerate all the numbers, but rather give a statistical summary in Table~\ref{table:cpe}. Notice that since the CTPs are in the 1-3\% range, the effective number of targeted nodes is correspondingly larger. 
We defer the numbers for \dblp and \livej to \textsection\ref{sec:scala}.
}

\WL{All experiments were run on a 64-bit RedHat Linux server with Intel Xeon 2.40GHz CPU and 65GB memory.
Our largest configuration is \livej with 20 ads, which effectively has $69M \cdot 20 = 1.4B$ edges;
this is comparable with \cite{tang14}, whose largest dataset has 1.5B edges (Twitter).
}

\begin{table}[t!]
\small
\centering
	\begin{tabular}{|c | c | c | c | c | c | c|}
		\hline
		 &  \multicolumn{3}{|c|}{ Budgets} & \multicolumn{3}{|c|} {CPEs} \\
 		 \hline
		 Dataset & mean & min  & max & mean & min & max  \\ \hline
		\flix & 375  & 200 & 600  & 5.5 & 5 & 6  \\ \hline
		\epi & 215 & 100  & 350  & 4.35 & 2.5 & 6 \\ \hline
	\end{tabular}
\caption{Advertiser budgets and cost-per-engagement values}
\label{table:cpe}
\vspace{-1mm}
\end{table}

\begin{figure*}[t!]
\begin{tabular}{cccc}
    \includegraphics[width=.24\textwidth]{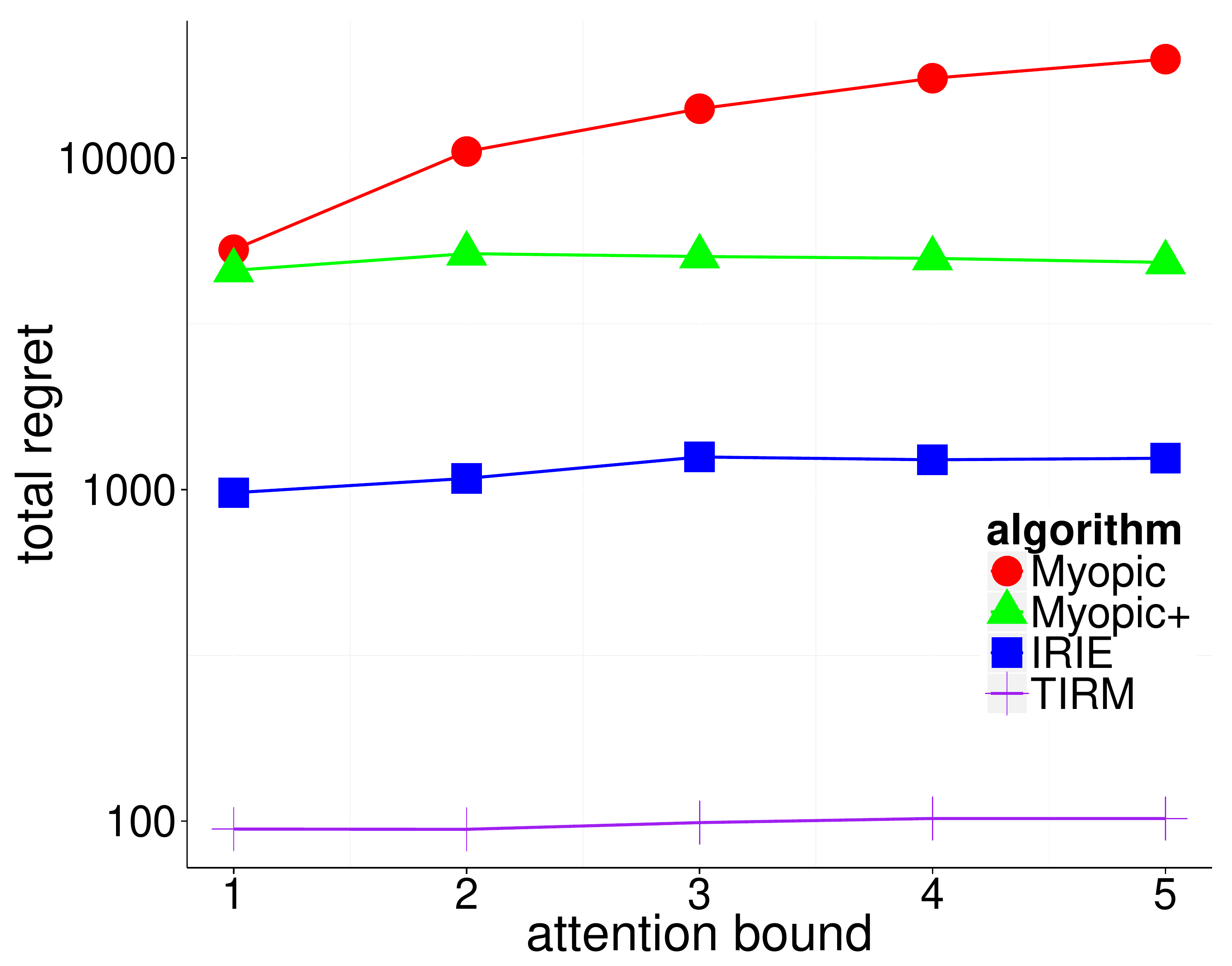}&
    \hspace{-2mm}\includegraphics[width=.24\textwidth]{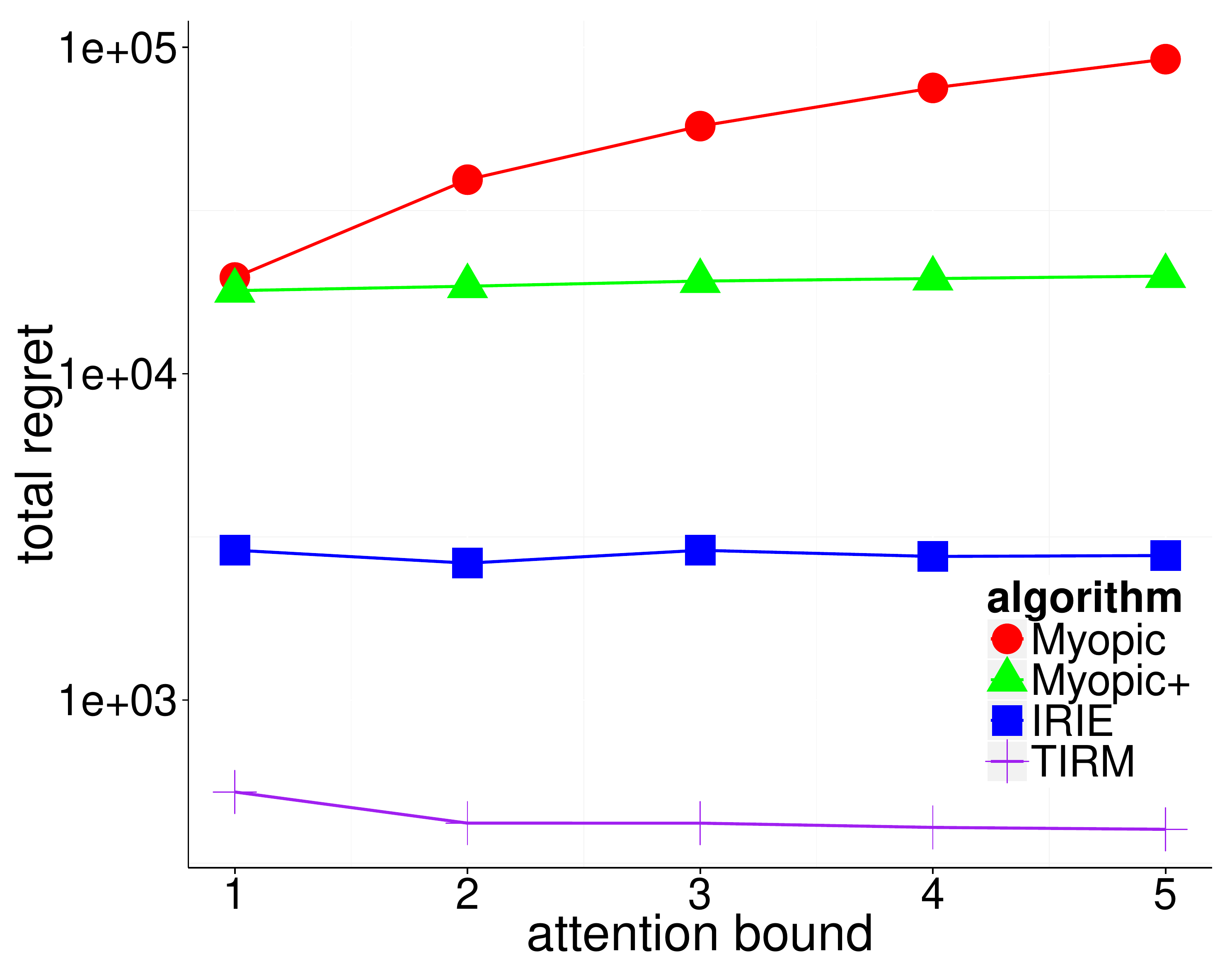}&
    \hspace{-2mm}\includegraphics[width=.24\textwidth]{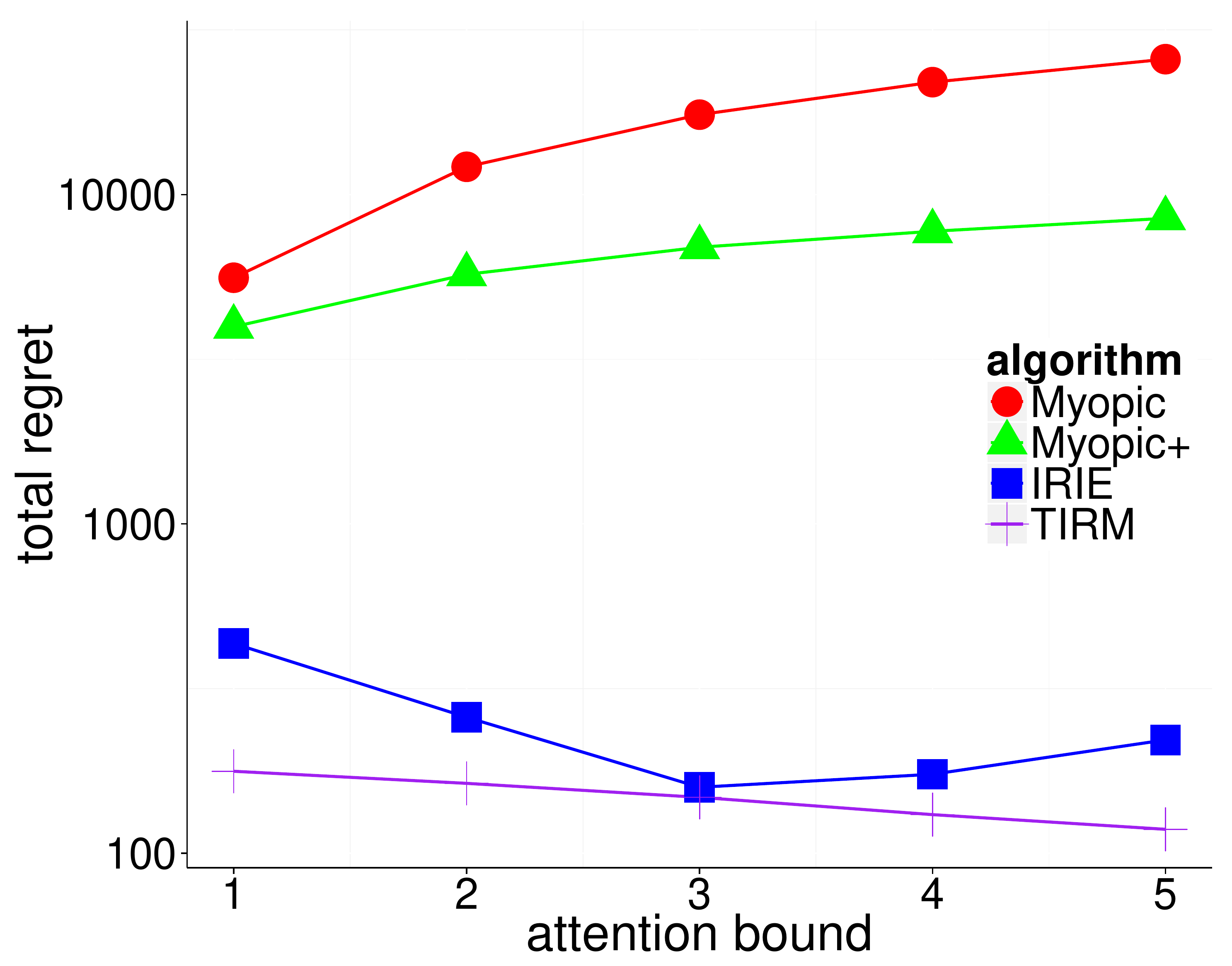}&
    \hspace{-2mm}\includegraphics[width=.24\textwidth]{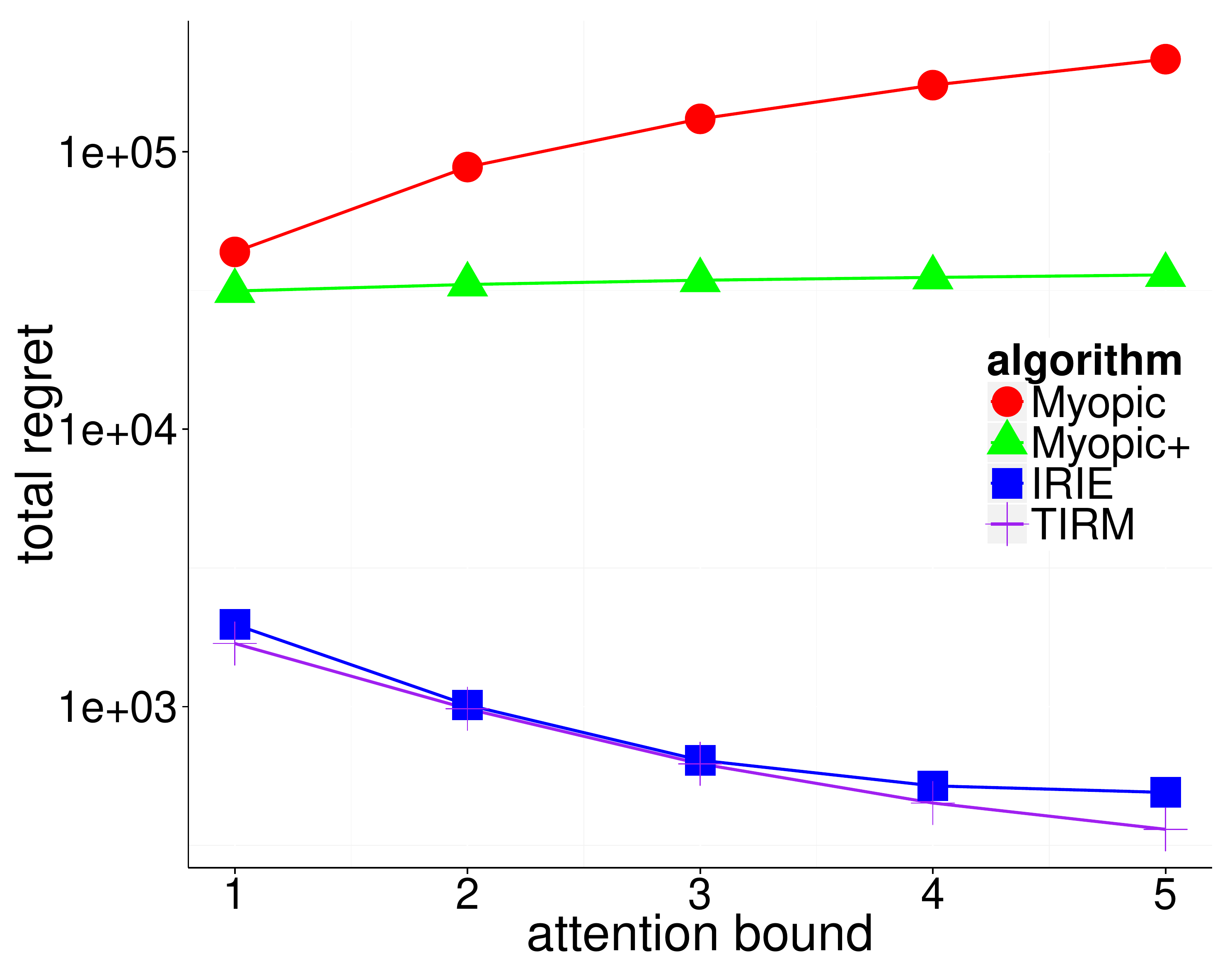}\\
	(a) \flix ($\lambda=0$)  & (b) \flix ($\lambda=0.5$) & (c) \epi ($\lambda=0$) & (d) \epi ($\lambda=0.5$)   \\
\end{tabular}
\vspace{-4mm}
\caption{Total regret (log-scale) vs. attention bound $\kappa_u$}
\label{fig:regret-kappa}
\end{figure*}

\begin{figure*}[t!]
\begin{tabular}{cccc}
    \includegraphics[width=.24\textwidth]{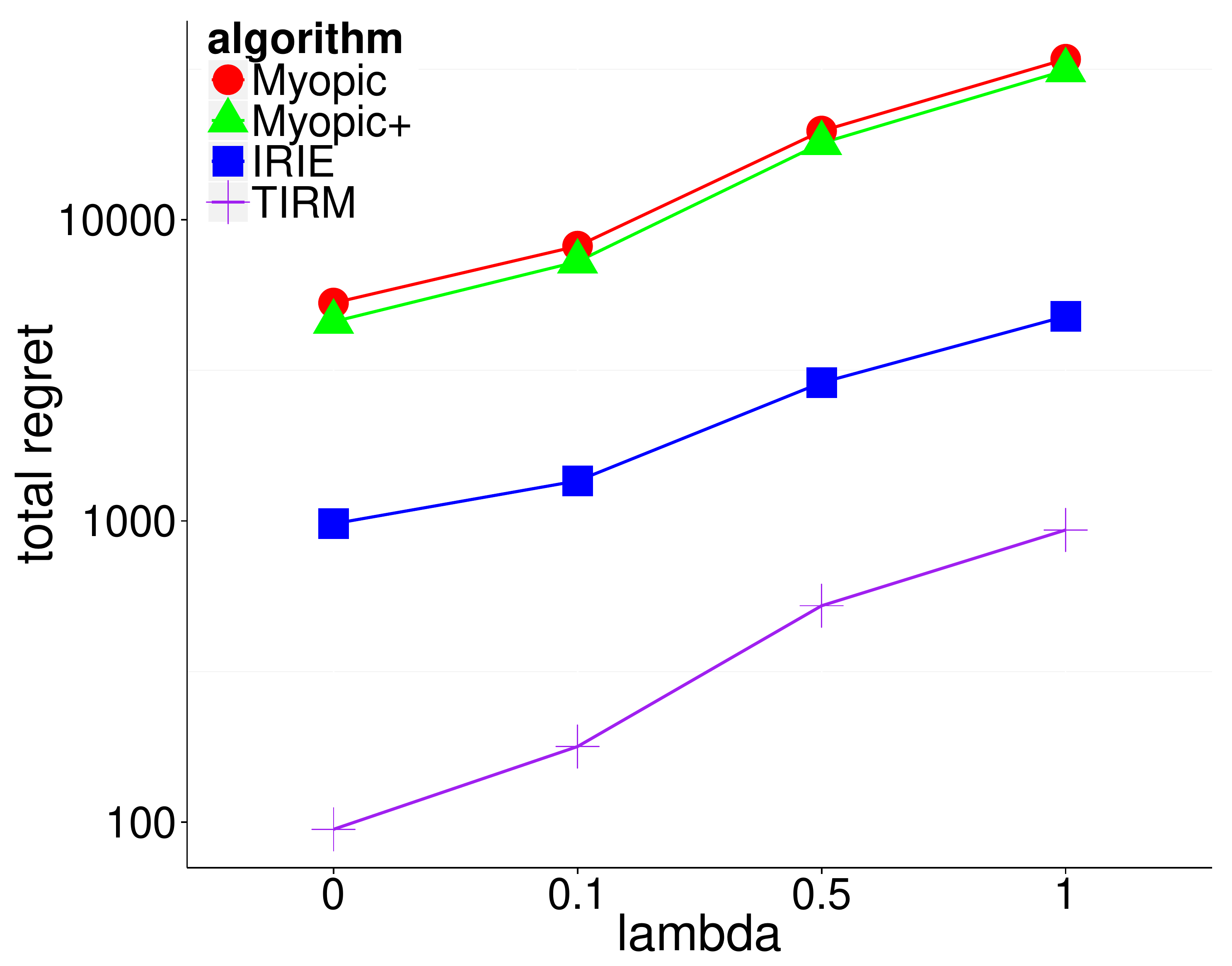}&
    \hspace{-2mm}\includegraphics[width=.24\textwidth]{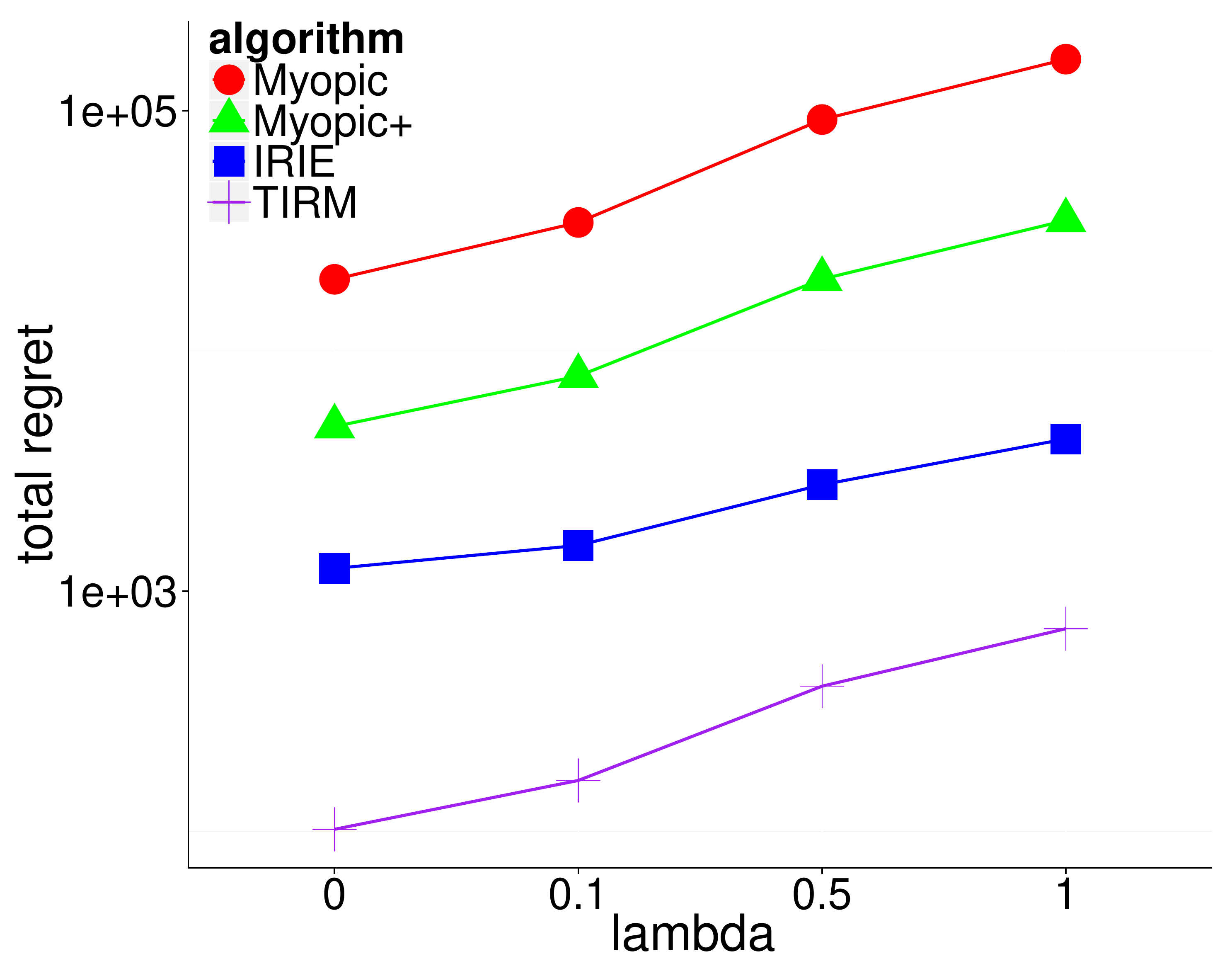}&
    \hspace{-2mm}\includegraphics[width=.24\textwidth]{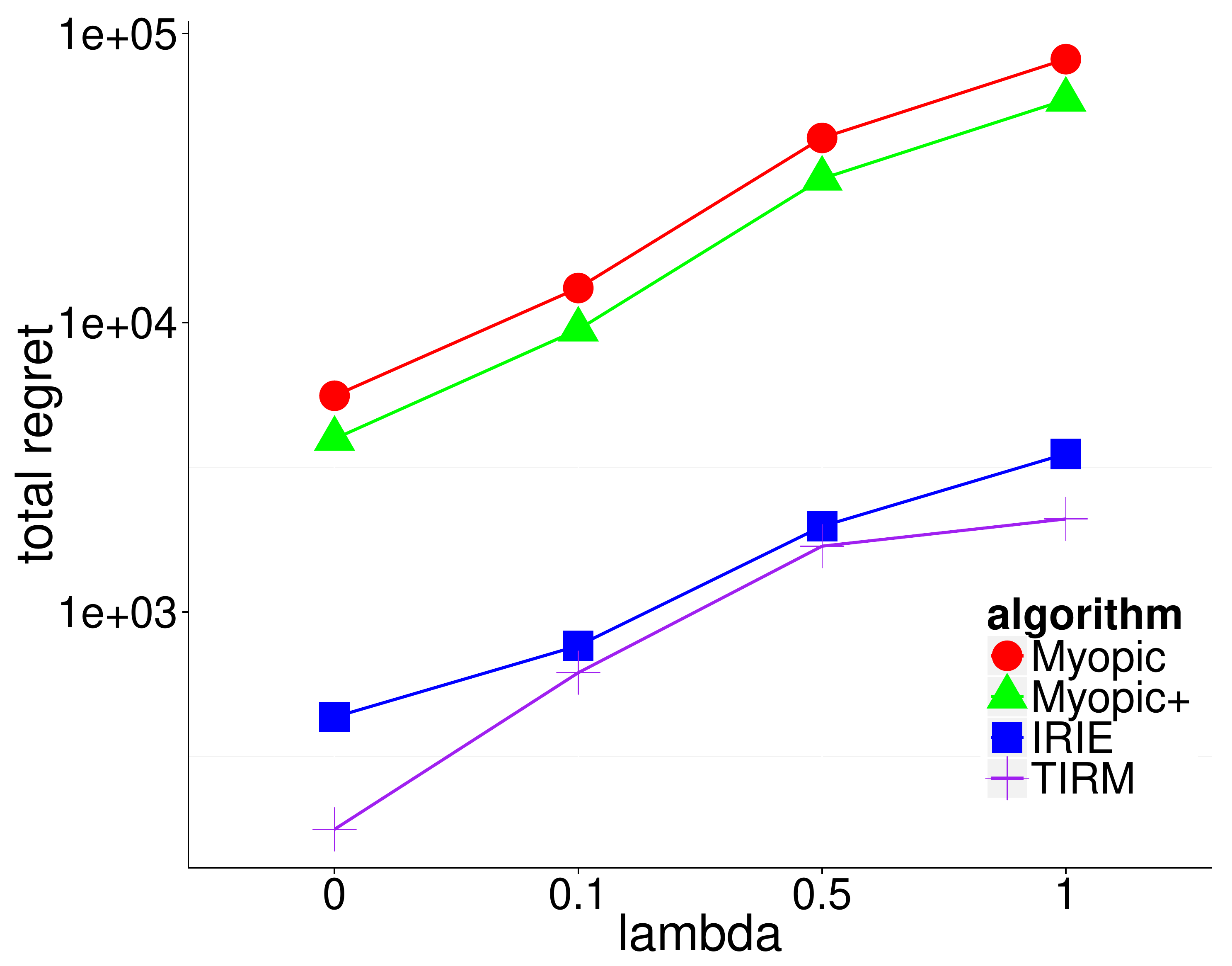}&
    \hspace{-2mm}\includegraphics[width=.24\textwidth]{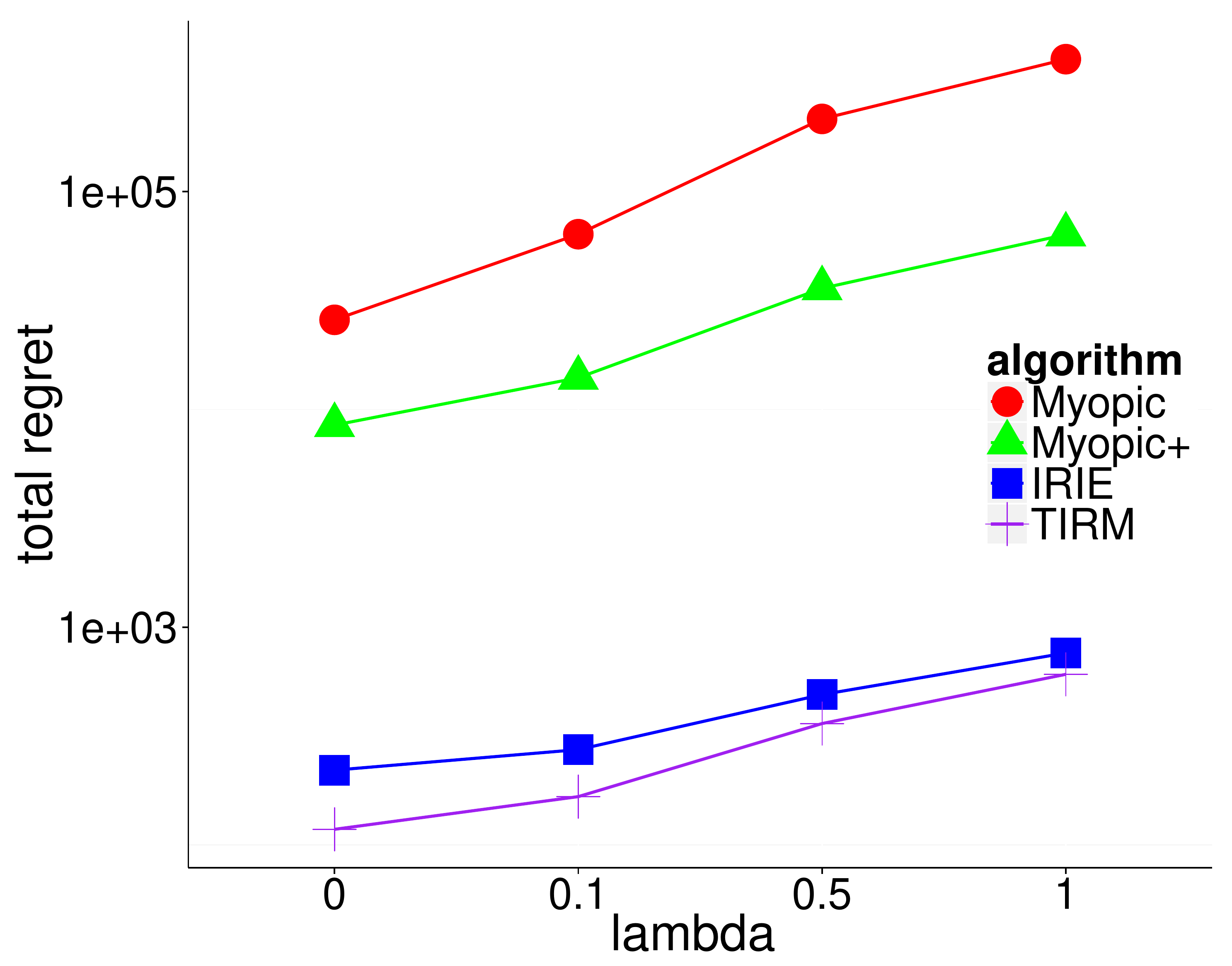}\\
	(a) \flix ($\kappa_u=1$)  & (b) \flix ($\kappa_u=5$) & (c) \epi ($\kappa_u=1$) & (d) \epi ($\kappa_u=5$)   \\
\end{tabular}
\vspace{-4mm}
\caption{Total regret (log-scale) vs. $\lambda$}
\label{fig:regret-lambda}
\end{figure*}

\smallskip\noindent\textbf{Algorithms.}
We test and compare the following four algorithms.

\smallskip\noindent$\bullet$ \myopicOne:
A baseline that assigns every user $u\in V$ in total $\kappa_u$  most relevant
ads $i$, i.e., those for which $u$ has the highest expected revenue, not considering any network effect, i.e., $\delta(u,i)\cdot cpe(i)$.
It is called ``myopic'' as it solely focuses on CTPs and CPEs and effectively ignores virality and budgets.
Allocation $\mathcal{A}$ in Fig.~\ref{fig:viral-ad-ex} follows this baseline.
	
\smallskip\noindent$\bullet$\myopicTwo:
This is an enhanced version of \myopicOne which takes budgets, but not virality, into account.
For each ad, it first ranks users w.r.t.\ CTPs and then selects seeds using this order until budget is exhausted.
User attention bounds are taken into account by going through the ads round-robin and advancing to the next seed if the current node $u$ is already assigned to $\kappa_u$ ads.

\smallskip\noindent$\bullet$ \irie:
An instantiation of Algorithm~\ref{alg:rmVanilla}, with the IRIE
heuristic~\cite{jung12} used for influence
spread estimation and seed selection.
IRIE has a damping factor $\alpha$ for accurately estimating
influence spread in its framework.
Jung et al.~\cite{jung12} 
report that $\alpha=0.7$ performs best on the datasets they tested. 
We did extensive testing on our datasets and found that $\alpha = 0.8$ gave the best spread estimation, and thus used $0.8$ in all quality experiments.


\smallskip\noindent$\bullet$ \fastAlgo: Algorithm~\ref{alg:rmRRsets}.
We set $\varepsilon$ to be $0.1$ for quality experiments on \flix and \epi, and $0.2$
for scalability experiments on \dblp and \livej (following \cite{tang14}).  

For all algorithms, we evaluate the final regret of their output seed sets using Monte Carlo simulations (10K runs) for neutral, fair, and accurate comparisons.






\subsection{Results of Quality Experiments}
 

\smallskip\noindent\textbf{Overall regret.}
First, we compare overall regret (as defined in Eq.~\eqref{eq:total_regret}) against attention bound $\kappa_u$, varied from $1$ to $5$, with two choices $0$ and $0.5$ for $\lambda$.
Fig.~\ref{fig:regret-kappa} shows that the overall regret (in log-scale) achieved by \fastAlgo and \irie are significantly lower than that of \myopicOne and \myopicTwo.
For example, on \flix with $\lambda = 0$ and $\kappa_u=1$, overall regrets of \fastAlgo, \irie, \myopicOne, and \myopicTwo, expressed relative to the total budget, are $2.5\%$, $26.1\%$, $122\%$, $141\%$, respectively.
On \epi with the same setting,  the corresponding regrets are $6.5\%$, $15.9\%$, $145\%$, and $205\%$.
\myopicOne, and \myopicTwo typically always overshoot the budgets as they are not virality-aware when choosing seeds. Notice that even though \myopicTwo is budget conscious, it still ends up overshooting the budget as a result of not factoring in virality in seed allocation. 
In almost all cases, overall regret by \fastAlgo goes down as $\kappa_u$ increases. 
The trend for \myopicOne and \myopicTwo is the opposite, caused by their larger overshooting with larger $\kappa_u$.
This is because they will select more seeds as $\kappa_u$ goes up, which causes higher revenue (hence regret) due to more virality.



\begin{figure}[htp]
\centering
\subfigure[\flix]{ 
\includegraphics[width=0.23\textwidth]{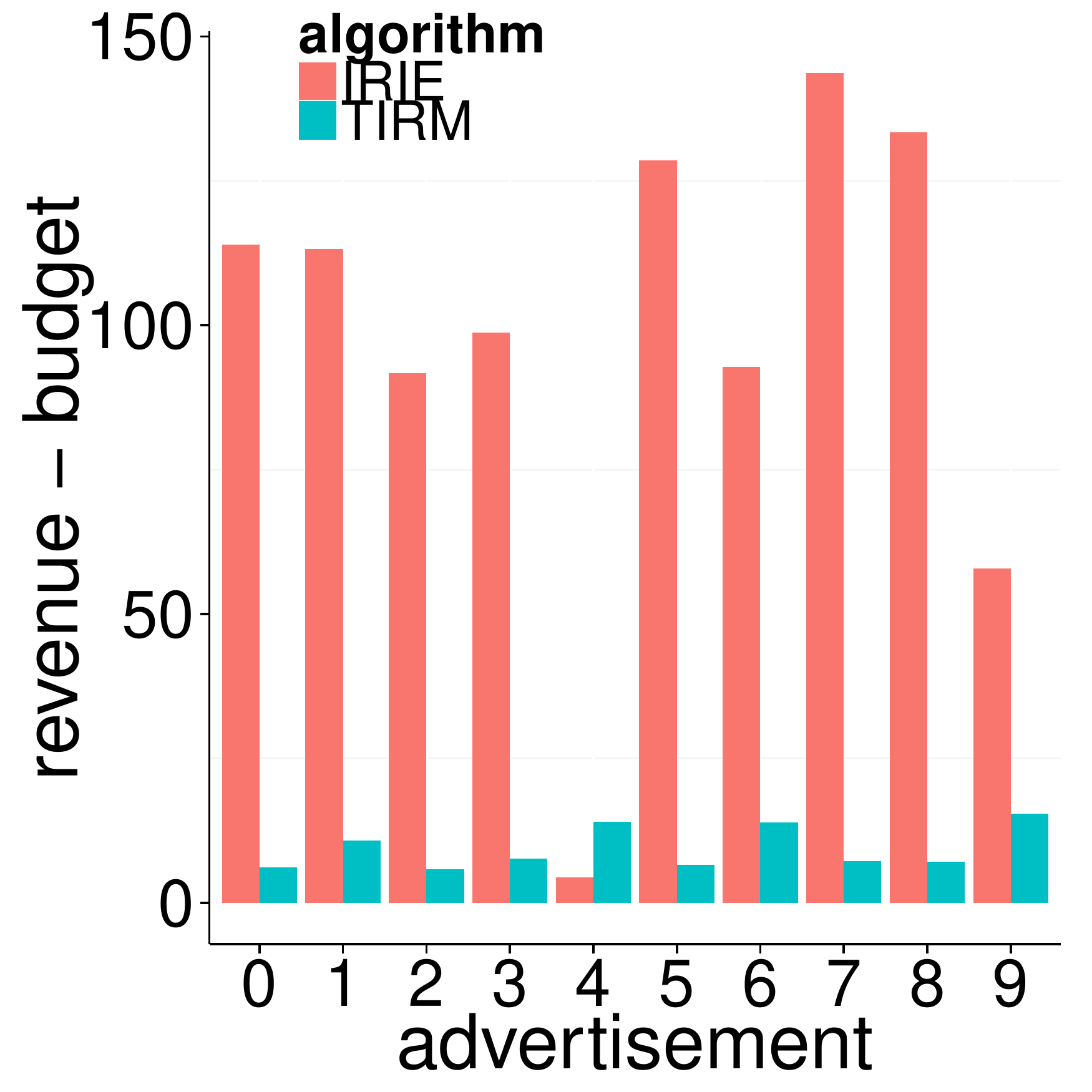}
}
\hspace{-3mm}
\subfigure[\epi]{
\includegraphics[width=0.23\textwidth]{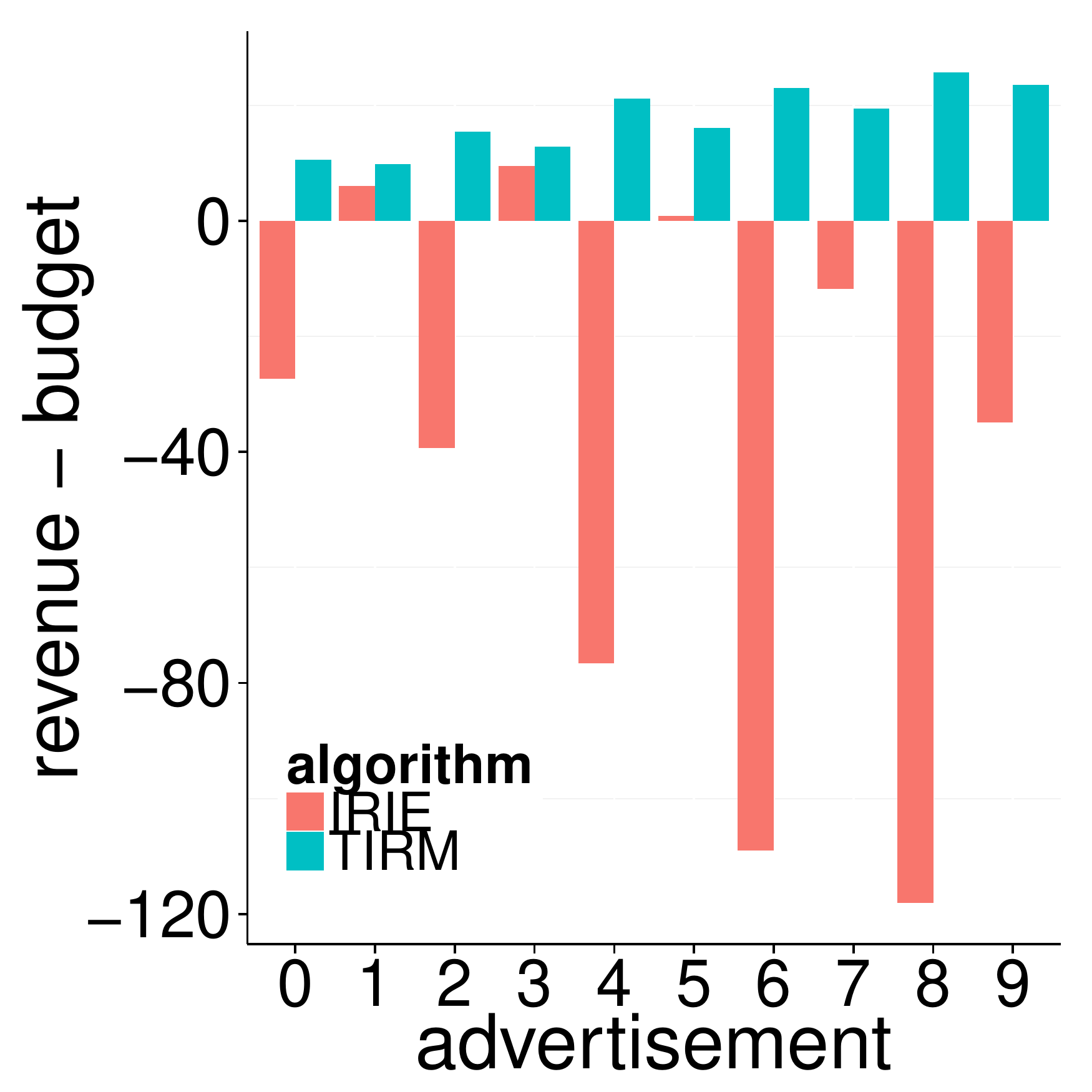}
}
\vspace{-2mm}
\caption{Distribution of individual regrets ($\lambda = 0$, $\kappa_u = 5$).}\label{fig:regret-distr}
\vspace{-2mm}
\label{fig:regret-distr}
\end{figure}

We also vary $\lambda$ to be $0$, $0.1$, $0.5$, and $1$ and show the overall regrets under those values in Fig.~\ref{fig:regret-lambda} (also in log-scale), with two choices $1$ and $5$ for $\kappa_u$.
As expected, in all test cases as $\lambda$ increases, the overall regret also goes up.
The hierarchy of algorithms (in terms of performance) remains the same as in Fig.~\ref{fig:regret-kappa}, with \fastAlgo being the consistent winner.
\WL{Note that even when $\lambda$ is as high as $1$, \fastAlgo still wins and performs well.
This suggests that the $\lambda$-assumption ($\lambda \le \delta(u,i)\cdot cpe(i)$, $\forall$ user $u$ and ad $i$) in Theorem~\ref{thm:full-regret} is conservative as \fastAlgo can still achieve relatively low regret even with large $\lambda$ values.}

\smallskip\noindent\textbf{Drilling down to individual regrets.}
Having compared overall regrets, we drill down into the budget-regrets (see \textsection\ref{sec:theory}) achieved for different individual ads by \fastAlgo and \irie.
Fig.~\ref{fig:regret-distr} shows the distribution of budget-regrets across advertisers for both algorithms.
On \flix, both algorithms overshoot for all ads, but the distribution of \fastAlgo-regrets is much more uniform than that of \irie-regrets.
E.g., for the fourth ad, \irie even achieves a smaller regret than \fastAlgo, but for all other ads, their \irie-regret is at least $3.8$ times as large as the  \fastAlgo-regret, showing a heavy skew. 
On \epi, \fastAlgo slightly overshoots for all advertisers as in the case of \flix, while \irie falls short on 7 out of 10 ads and its budget-regrets are larger than \fastAlgo for most advertisers.
\WL{
Note that \myopicOne and \myopicTwo are not included here as Figs.~\ref{fig:regret-kappa} and \ref{fig:regret-lambda} have clearly demonstrated that they have significantly higher overshooting\footnote{Their regrets are all from overshooting the budget on account of ignoring virality effects.}.
}



\smallskip\noindent\textbf{Number of targeted users.}
We now look into the distinct number of nodes targeted at least once by each  algorithm, as $\kappa_u$ increases from $1$ to $5$.
Intuitively, as $\kappa_u$ decreases, each node becomes ``less available'', and thus we may need more distinct nodes to cover all budgets, causing this measure to go up.
The stats in Table~\ref{table:numseeds-flix} confirm this intuition, in the case of \fastAlgo, \irie, and \myopicTwo.
\myopicOne is an exception since it allocates an ad to every user (i.e., all $|V|$ nodes are targeted). 
\WL{
Note that on \epi, \fastAlgo \LL{targeted} more nodes than \irie.
The reason is that \irie tends to overestimate influence spread on \epi, resulting in pre-mature termination of Greedy.
When MC is used to estimate ground-truth spread, the revenue would fall short of budgets (see Fig.~\ref{fig:regret-distr}).
The behavior of \irie is completely the opposite on \flix, showing its lack of consistency as a pure heuristic.
}

\begin{table}
\centering
\small
  \begin{tabular}{|c|c|c|c|c|c|}
    \hline
     {\bf \flix} & $\kappa_u = 1$ & $2$ &$3$ & $4$ & $5$ \\ \hline
    \fastAlgo & 868 & 352 & 319 & 263 & 257  \\ \hline
    \irie & 3.7K & 1.7K & 1.5K & 1237 & 1222 \\ \hline
    \myopicOne & 29K & 29K & 29K & 29K & 29K  \\ \hline
    \myopicTwo & 27K & 13K & 9.6K & 7.5K & 6.6K  \\ \hline
    \hline
    {\bf \epi} & $\kappa_u = 1$ & $2$ &$3$ & $4$ & $5$ \\ \hline
	\fastAlgo & 4.4K & 901 & 396 & 233 & 175  \\ \hline
    \irie & 3.1K & 826 & 393 & 251 & 183 \\ \hline
    \myopicOne & 76K & 76K & 76K & 76K & 76K  \\ \hline
    \myopicTwo & 55K & 28K & 19K & 15K & 13K  \\ \hline
  \end{tabular}
 \caption{Number of nodes targeted vs. attention bounds ($\lambda = 0$)}
 \label{table:numseeds-flix}
 \vspace{-4mm}
\end{table}

\begin{figure*}
\begin{tabular}{cccc}
    \includegraphics[width=.21\textwidth]{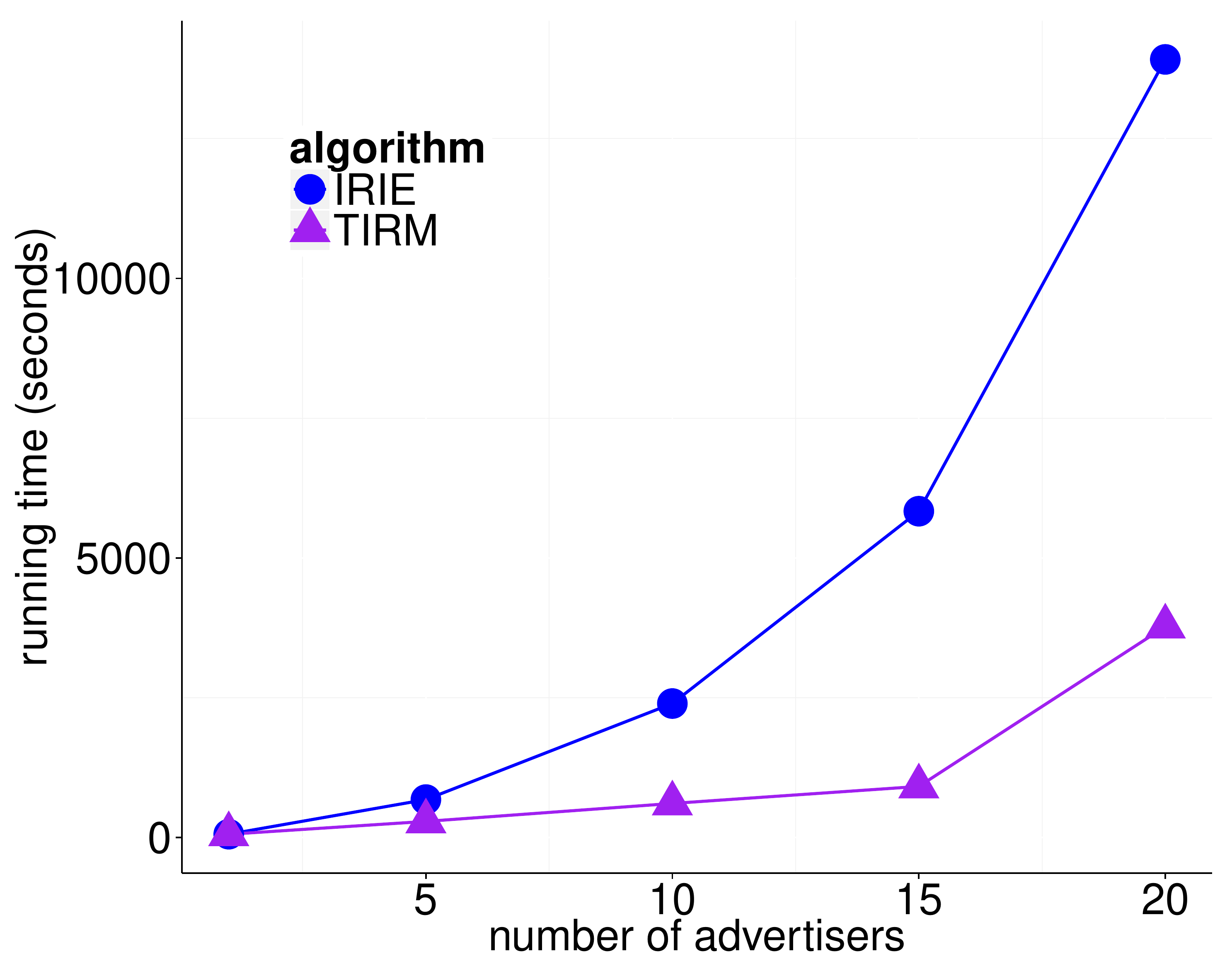}&
    \hspace{2mm}\includegraphics[width=.21\textwidth]{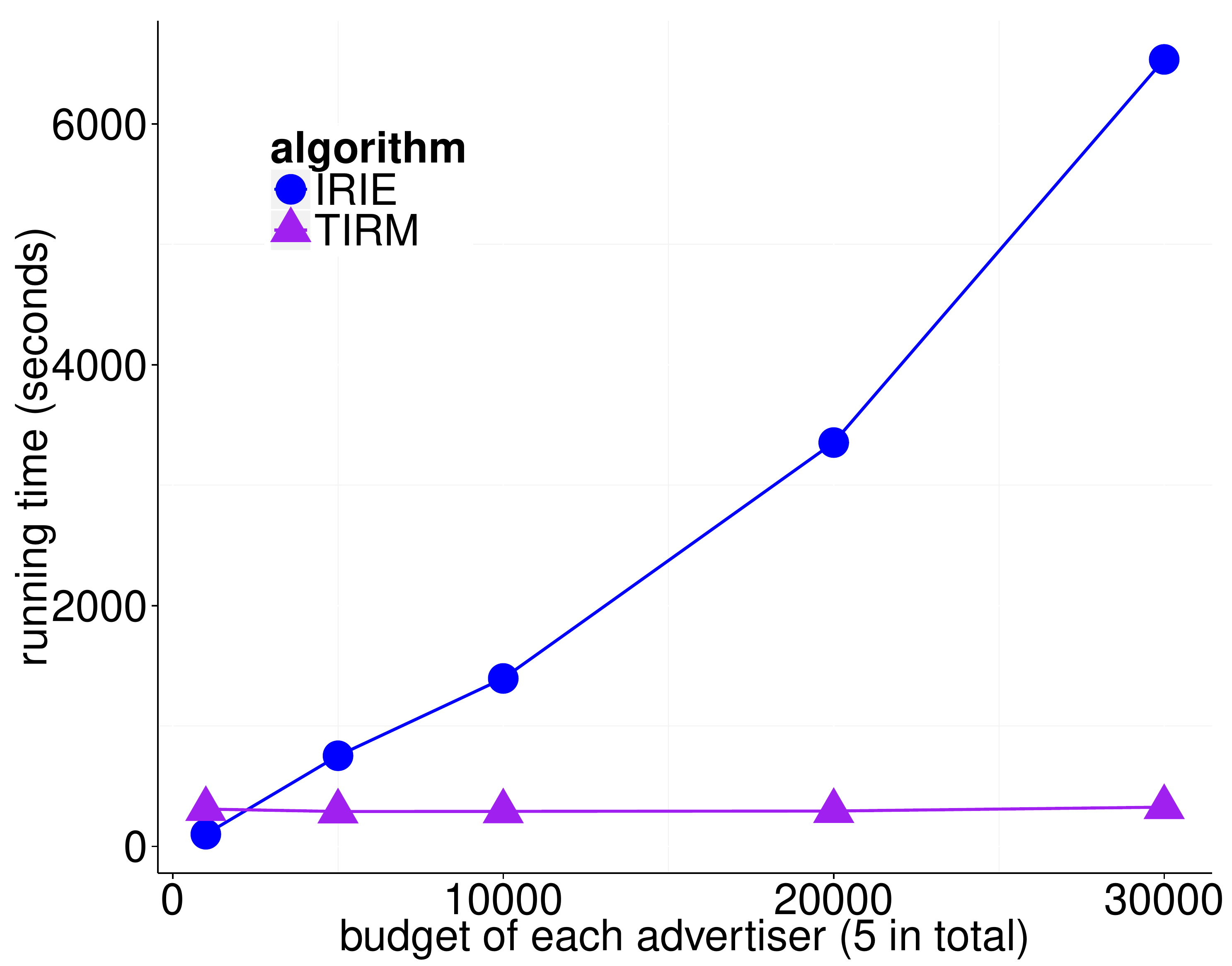}&
    \hspace{2mm}\includegraphics[width=.21\textwidth]{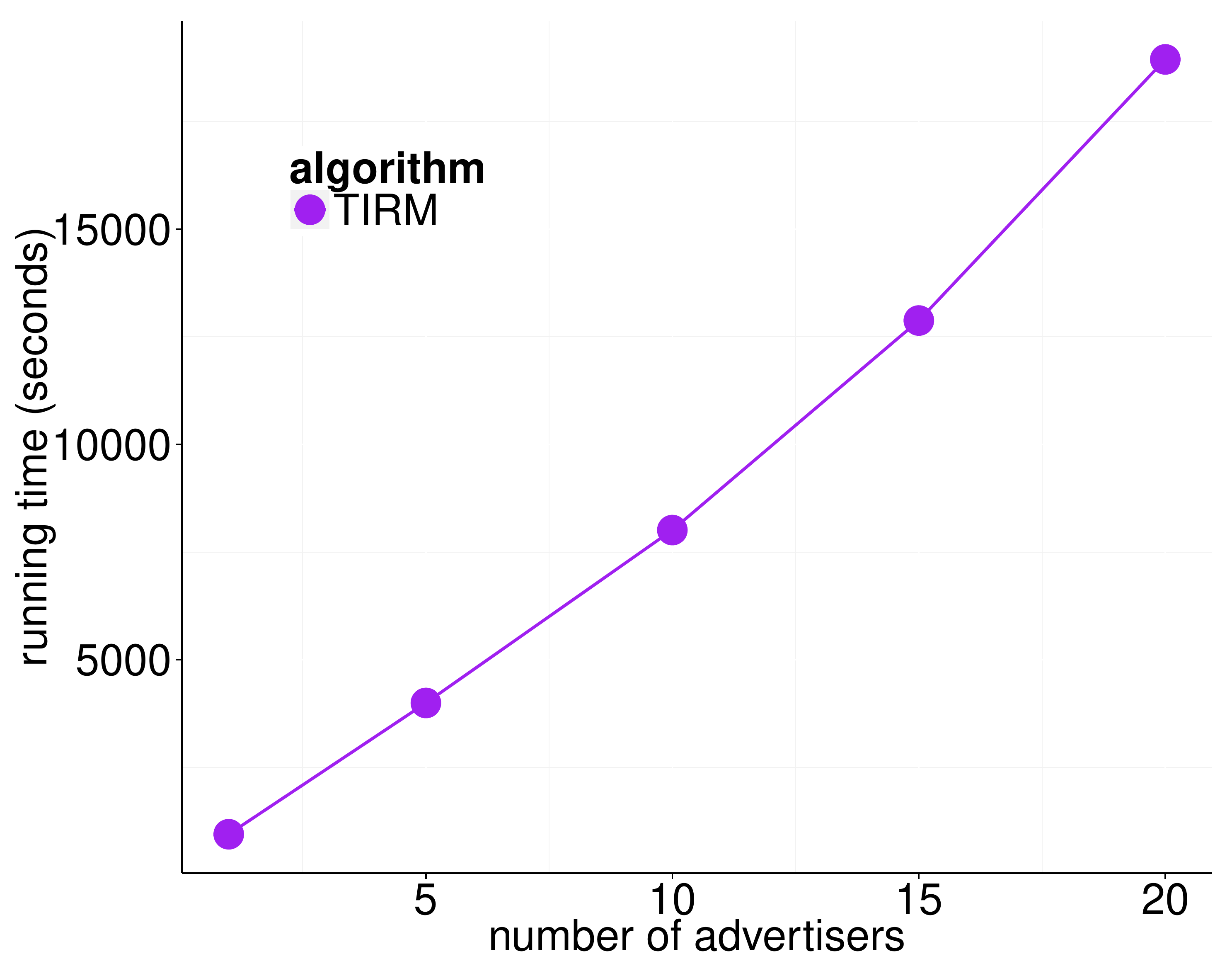}&
    \hspace{2mm}\includegraphics[width=.21\textwidth]{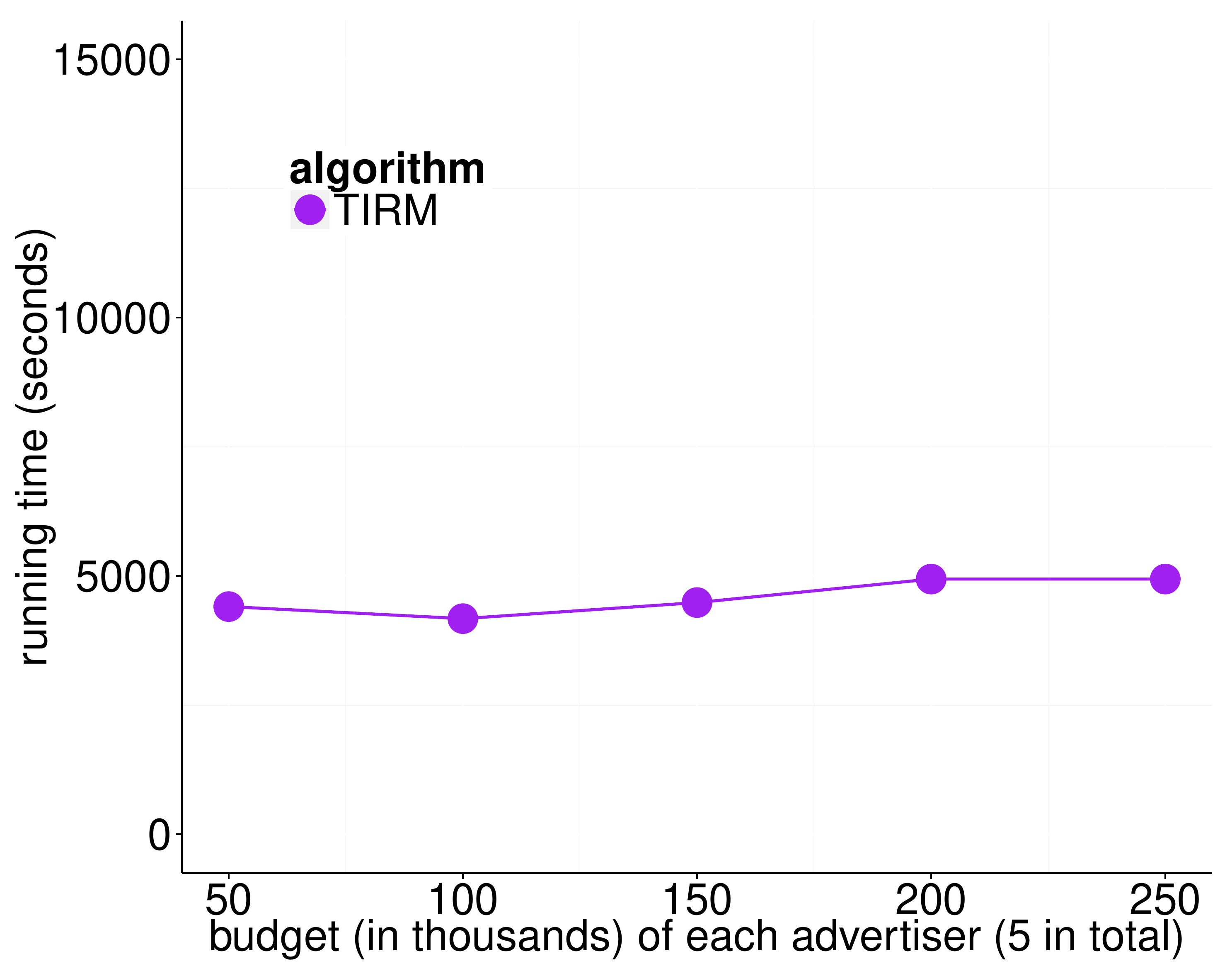}\\
	(a) \dblp ($h$) & (b) \dblp (budgets)  & (c) \livej ($h$) & (d) \livej (budgets)  \\
\end{tabular}
\vspace{-2mm}
\caption{Running time of \fastAlgo and \irie on \dblp and \livej}
\label{fig:time}
\end{figure*}

\subsection{Results of Scalability Experiments}\label{sec:scala}
We test the scalability of \fastAlgo and \irie on \dblp and \livej.
For simplicity, we set all CPEs and CTPs to $1$ and $\lambda$ to $0$, and the values of these parameters do not affect running time or memory usage. 
Influence probabilities on each edge $(u,v)\in E$ are computed using the Weighted-Cascade model~\cite{ChenWW10}: $p_{u,v}^i = \frac{1}{|N^{in}(v)|}$ for all ads $i$.
\WL{We set $\alpha=0.7$ for \irie and $\varepsilon = 0.2$ for \fastAlgo, in accordance with the settings in~\cite{tang14,jung12}}.
Attention bound $\kappa_u = 1$ for all users. 
We emphasize that our setting is fair and ideal for testing scalability as it simulates a fully competitive case: all advertisers compete for the same set of influential users (due to all ads having the same distribution over the topics) and the attention bound is at its lowest, which in turn will ``stress-test'' the algorithms by prolonging the seed selection process.

\CA{We test the running time of the algorithms in two dimensions: Fig.~\ref{fig:time}(a) \& \ref{fig:time}(c) vary $h$ (number of ads) with per-advertiser budgets $B_i$ fixed (5K for \dblp, 80K for \livej), while Fig.~\ref{fig:time}(b) \& \ref{fig:time}(d) vary $B_i$ when fixing $h=5$. Note that \irie results on \livej (Fig.~\ref{fig:time}(c) \& \ref{fig:time}(d)) are excluded due to its huge running time, details to follow.}

At the outset, notice that \fastAlgo significantly outperforms \irie in terms of running time. 
\CA{Furthermore, as shown in Fig.~\ref{fig:time}(a), the gap between \fastAlgo and \irie on \dblp becomes larger as $h$ increases.
For example, when $h=1$,  both algorithms finish in 60 secs, but when $h=15$, \fastAlgo is 6 times faster than \irie.} 
\eat{On \livej, \fastAlgo took about 5 hours to complete while \irie could not finish in 48 hours (so we exclude it from \livej plots). } 

\CA{On \livej, \fastAlgo scales almost linearly w.r.t. the number of advertisers,
It took about 16 minutes with $h=1$ (47 seeds chosen) and 5 hours with $h=20$ ($4649$ seeds).
\irie took about 6 hours to complete for $h=1$, and did not finish after 48 hours for $h \ge 5$. 
When budgets increase (Fig.~\ref{fig:time}(b)), \irie's time will go up (super-linearly) due to more iterations of seed selections, but \fastAlgo remains relatively stable (barring some minor fluctuations).
On \livej , \fastAlgo took less than 75 minutes with $B_i = 50K$ ($254$ seeds). 
Note that once $h$ is fixed, \fastAlgo's running time depends heavily on the required number of random RR-sets ($\theta$) for each advertiser rather than budgets, as seed selection is a linear-time operation for a given sample of RR-sets.
Thus, the relatively stable trend on Fig.~\ref{fig:time}(b) \& \ref{fig:time}(d) 
is due to the subtle interplay among the variables to compute $L(s, \varepsilon)$ (Eq.~\ref{eqn:timLB}); similar observations were made for TIM in \cite{tang14}.} 

%
%
%
%
%
%

Table~\ref{table:mem} shows the memory usage of \fastAlgo and \irie.
As \fastAlgo relies on generating a large number of random RR-sets for accurate estimation of influence spread, we observe high memory consumption by this algorithm, similar to the TIM algorithm~\cite{tang14}. The usage steadily increases with $h$. The memory usage of \irie is modest, as its computation requires merely the input graph and probabilities. However, \irie is a heuristic with no guarantees, which is reflected in its relatively poor regret performance compared to \fastAlgo. Furthermore, as seen earlier, \fastAlgo scales significantly better than \irie on all datasets.


\begin{table}
\centering
\small
  \begin{tabular}{|c|c|c|c|c|c|}
    \hline
    {\bf \dblp} & $h = 1$ & $5$ & $10$ & $15$ & $20$ \\ \hline
    \fastAlgo & 2.59 & 12.6 & 27.1 & 40.6 & 60.8  \\ \hline
    \irie & 0.16 & 0.30 & 0.48 & 0.54 & 0.84 \\ \hline
    \hline
    {\bf \livej} & $h = 1$ & $5$ & $10$ & $15$ & $20$ \\ \hline
    \fastAlgo & 3.72 & 15.6 & 32.5 & 47.7 & 60.9  \\ \hline
  \end{tabular}
 \caption{Memory usage (GB)}
 \label{table:mem}
 \vspace{-6mm}
\end{table}

\section{Conclusions and Future Work}
\label{sec:concl}

In this work, we build a bridge between viral marketing and social advertising, by 
drawing on the viral marketing literature to study influence-aware ad allocation for social advertising,
under real-world business model, paying attention to 
important practical factors like relevance, social proof, user attention bound, and advertiser budget.
In particular, we study the problem of regret minimization from the host perspective, characterize its hardness and devise a simple scalable algorithm with quality guarantees w.r.t.\ the total budget.
Through extensive experiments we 
demonstrate its superior performance over natural baselines. 

Our work takes a first step toward enriching the framework of social advertising by integrating it with powerful ideas from viral marketing and making the latter more applicable to real online marketing problems. It opens up several interesting avenues for further research. 
\eat{We assumed a discrete-time propagation model, and that the social network and peer-to-peer influence probabilities are known beforehand. Furthermore, we did not enforce any \emph{hard} competition constraints (e.g., a user will buy at most one camera).} 
Studying continuous-time propagation models, possibly with the network and/or influence probabilities not known beforehand (and to be learned), and possibly in presence of hard competition constraints, is a direction that offers a wealth of possibilities for future work. 




\end{document}